\documentclass[aps, pra, 10pt, superscriptaddress, twocolumn, longbibliography, floatfix]{revtex4-2}

\usepackage[titletoc]{appendix} 

\usepackage[nointlimits]{amsmath}
\usepackage{graphicx,nicefrac}
\usepackage{subfigure}
\usepackage{color}
\usepackage{txfonts}
\usepackage{mathtools}
\usepackage{hyperref}
\usepackage{tabularx}
\hypersetup{colorlinks=true, citecolor=blue, linkcolor=blue}

\usepackage{amsfonts,amsmath,amssymb}
\usepackage{amsthm}
\usepackage{array,bm,color}
\usepackage{epsfig,graphicx,nomencl,revsymb4-1,upgreek,url}
\usepackage{hyperref}
\usepackage{algpseudocode}
\usepackage{xspace}
\usepackage{tabularx}
\usepackage[normalem]{ulem}
\usepackage{changes}
\usepackage{enumitem}
\usepackage{booktabs}
\usepackage{mathtools}
\usepackage{lipsum}
\usepackage{mwe}

\usepackage{comment}
\usepackage{multirow}

\def\bra#1{{\left\langle #1 \right|}}
\def\ket#1{{\left| #1 \right\rangle}}
\def\sbra#1{{\left\langle\langle #1 \right|}}
\def\sket#1{{\left| #1 \right\rangle\rangle}}

\newtheorem{lemma}{Lemma}

\theoremstyle{definition}

\newcommand{\algiers}{\texttt{ibmq\_algiers}}

\DeclareMathOperator*{\E}{\mathbb{E}}
\newcommand{\Tr}{\text{Tr}}

\newcommand{\Q}[1]{\texttt{Q#1}}

\newcommand{\cnot}{\mathsf{cnot}}
\newcommand{\mcm}{\mathsf{MCM}}

\newcommand{\eomega}{\varepsilon_\Omega}
\newcommand{\romega}{r_\Omega}

\newcommand{\Hunmeas}{\mathbb{H}_{\mathrm{unmeas}}}
\newcommand{\Hmeas}{\mathbb{H}_{\mathrm{meas}}}

\newcommand{\Xa}[1]{X^{\otimes #1}}
\newcommand{\Za}[1]{Z^{\otimes #1}}

\newcommand{\spre}[1]{s^{\mathrm{pre}}_{#1}}
\newcommand{\spost}[1]{s^{\mathrm{post}}_{#1}}

\newcommand{\sunmeas}[1]{s^{\mathrm{unmeas}}_{#1}}

\newcommand{\panti}[1]{p_{\mathrm{anti}}(#1)}
\newcommand{\ptrans}[1]{p^{\mathrm{trans}}_{#1}}
\newcommand{\Tab}[2]{\mathcal{T}_{#1,#2}}

\newcommand{\customsection}[1]{\vspace{0.2cm}\noindent\emph{#1}.}

\begin{document}
\title{Measuring error rates of mid-circuit measurements}
\author{Daniel Hothem}
\thanks{dhothem@sandia.gov}
\affiliation{Quantum Performance Laboratory, Sandia National Laboratories, Livermore, CA 94550}
\author{Jordan Hines}
\affiliation{Quantum Performance Laboratory, Sandia National Laboratories, Livermore, CA 94550}
\affiliation{Department of 
Physics, University of California, Berkeley, CA 94720}
\author{Charles Baldwin}
\author{Dan Gresh}
\affiliation{Quantinuum, 303 S. Technology Ct., Broomfield, CO 80021}
\author{Robin Blume-Kohout}
\affiliation{Quantum Performance Laboratory, Sandia National Laboratories, Albuquerque, NM 87185}
\author{Timothy Proctor}
\affiliation{Quantum Performance Laboratory, Sandia National Laboratories, Livermore, CA 94550}
\begin{abstract} 
High-fidelity mid-circuit measurements, which read out the state of specific qubits in a multiqubit processor without destroying them or disrupting their neighbors, are a critical component for useful quantum computing. They enable fault-tolerant quantum error correction, dynamic circuits, and other paths to solving classically intractable problems. But there are almost no methods to assess their performance comprehensively. We address this gap by introducing the first randomized benchmarking protocol that measures the rate at which mid-circuit measurements induce errors in many-qubit circuits.  Using this protocol, we detect and eliminate previously undetected measurement-induced crosstalk in a 20-qubit trapped-ion quantum computer. Then, we use the same protocol to measure the rate of measurement-induced crosstalk error on a 27-qubit IBM Q processor, and quantify how much of that error is eliminated by dynamical decoupling.
\end{abstract}
\maketitle

Quantum computers promise to solve classically intractable problems~\cite{Sho97, Rub23, Har09, Cao19, proctor2024benchmarking}. However, they are unlikely to do so by running quantum programs consisting solely of reversible unitary logic operations (``gates'') \cite{gidney2021factor, Rub23, lee2021even}. Gates are noisy, and their errors accumulate and obscure the answer.  All promising avenues to quantum computational advantage rely critically on \textit{mid-circuit measurements} (MCMs) to reduce entropy during program execution. The best-known (and most trusted) path is fault-tolerant quantum computing using quantum error correction (QEC)~\cite{Ter15, Pou05,Blu23, Goo23, Wan23}, but other approaches like dynamic quantum circuits~\cite{decross2023qubit} also rely on high-fidelity MCMs. But whereas errors in reversible gates can be measured and quantified using a wide range of well-tested protocols~\cite{Emerson2005-fd,Emerson2007-am,Knill2008-jf, Magesan2011-hc, Mag12, Pro19, Pro22, Hin23, Hin23-2, proctor2021measuring, Polloreno2023-xa, Fla11, Erhard2019-wk, Harper2020-te, Nielsen2021-nu, Blume-Kohout2017-no}, almost no methods for quantifying errors in MCMs are available~\cite{proctor2024benchmarking}. Those proposed so far~\cite{Rud22, Str22, Pereira2023-nw,  Wag20, Gom16, Gae21, Gov22, Hines2024-qj, Zhang2024-zp, shirizly2024randomizedbenchmarkingprotocoldynamic} have limited scope.

To address this gap, we introduce and demonstrate the first fully scalable protocol for quantifying and measuring the combined rate of \textit{all} errors in MCM operations. Our protocol integrates MCMs into the canonical framework of randomized benchmarking (RB)~\cite{Emerson2005-fd,Emerson2007-am,Knill2008-jf, Magesan2011-hc, Mag12, Pro19, Pro22, Hin23, Hin23-2}, and models noisy measurement operations as general \textit{quantum instruments} \cite{Davies1970-js, Rud22}. MCM operations are placed naturally within random circuits, making it straightforward to measure the rate at which they cause errors in those circuits.  We call the protocol \textit{quantum instrument randomized benchmarking} (QIRB) because it captures general MCM errors described by quantum instruments \cite{Davies1970-js} and yields an overall rate of MCM errors that can be directly compared to RB-derived gate error rates found throughout the experimental quantum computing literature.  


\begin{figure*}[t!]
\includegraphics[width=\textwidth]{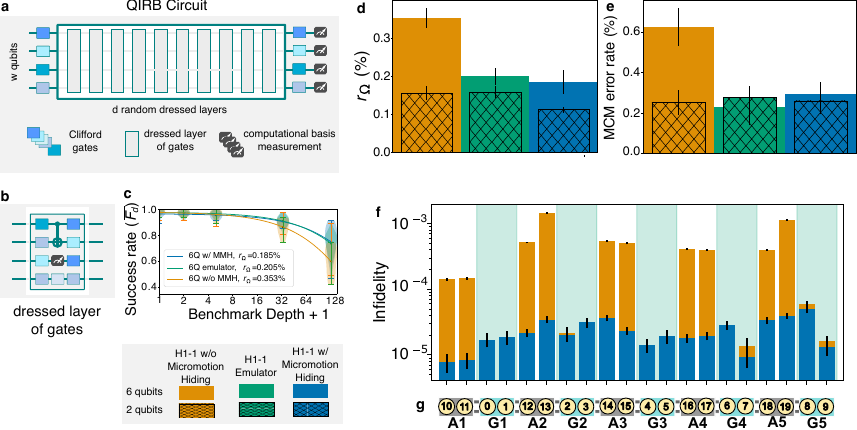}
\caption{\textbf{Quantifying measurement-induced crosstalk in a trapped-ion quantum computer using QIRB}. QIRB uses the MCM-containing random circuits shown in (a-b) to measure the average error rate ($r_{\Omega}$) of MCM-containing circuit layers. QIRB retains the simplicity and elegance of RB methods for measuring gate errors, as it estimates $r_{\Omega}$ as the decay rate of the average circuit success rates ($\overline{F}_d$) versus circuit depth, as shown in the examples of (c). We used QIRB to study MCMs in Quantinuum's H1-1 system, depicted in (g), which arranges 20 $\textsuperscript{171}\mathrm{Yb}^{+}$ ions into auxiliary (``A'') and gate (``G'') zones. Our initial experiments on H1-1 applied micromotion hiding \cite{Gae21} only to unmeasured ions in gate zones during MCMs, as ions in auxiliary zones are distant from measured ions. We observed higher error rates in six-qubit QIRB than predicted [compare left with center bars in (d) and (e)]. We added micromotion hiding for all unmeasured ions, reran QIRB, and observed that the six-qubit error rate was approximately halved [right bars in (d) and (e)] and was now consistent with the predictions of Quantinuum's emulator. Results from bright-state depumping experiments (f) independently confirmed the existence of long-range, MCM-induced crosstalk and its mitigation by extra micromotion hiding. These experiments quantify MCM-induced crosstalk errors by measuring the rate at which unmeasured ions, prepared in the $\ket{1}$ state, leak out of the computational space as gate-zone ions are measured.}\label{fig:qirb-combined}
\end{figure*}

We demonstrate QIRB's validity and practical power by using it to probe, understand, and reduce MCM errors in two very different quantum processors. First, we applied QIRB to study MCMs in a 20-qubit Quantinuum system (H1-1) that was previously used to demonstrate high-performance QEC~\cite{Ryan-Anderson2022-un}. The results (Fig.~\ref{fig:qirb-combined}) revealed unexpectedly high rates of MCM-induced crosstalk errors that other methods failed to detect, and allowed us to deduce their location and cause.  After we implemented a targeted intervention---micromotion hiding~\cite{Gae21} in the system's auxiliary zones---we ran QIRB again, and found the errors had been eliminated. Second, we used QIRB to demonstrate that dynamical decoupling \cite{viola1999dynamical} reduces MCM-induced errors in $15$-qubit layers on IBM Q's 27-transmon \algiers\ system (Fig.~\ref{fig:ibmq-exp-results}). These experiments required a protocol that was both scalable (i.e., could characterize MCMs \textit{in situ} in many-qubit circuits) and generic (i.e., predictably sensitive to all kinds of MCM errors).  These same properties enable measuring and predicting the impact of noisy MCMs on fault-tolerant circuits.

Most previous work on characterizing measurement errors has focused on circuit-terminating measurements that can be described by positive operator valued measures (POVMs) \cite{Bisio2009-qg, Lundeen2009-jo,Feito2009, Maciejewski2020-vc, Nielsen2021-nu, Blume-Kohout2017-no}. Investigations of MCMs have emphasized tomographic reconstruction of quantum instruments~\cite{Rud22, Str22, Pereira2023-nw}, which enables detailed modeling of one- and two-qubit dynamics but not \textit{in situ} characterization of errors in large processors. Another approach, \textit{certifying} that a measurement is near-perfect or possesses certain properties~\cite{Wag20, Gom16}, can rule out specific errors but does not quantify the experimental impact of errors. Recent papers have probed MCM-induced crosstalk~\cite{Gae21, Gov22}, explored Pauli noise learning for MCMs \cite{Hines2024-qj, Zhang2024-zp}, and assessed single-qubit feedforward operations~\cite{shirizly2024randomizedbenchmarkingprotocoldynamic}, but do not attempt to detect or measure all the MCM errors that can appear \emph{in situ} in many-qubit circuits.  We developed QIRB specifically to overcome these limitations, enabling scalable quantification of all MCM errors \textit{in situ}.

RB is the most commonly used \textit{gate} characterization protocol because it mixes all kinds of gate errors together to yield a single, simple error rate~\cite{Emerson2005-fd,Emerson2007-am,Knill2008-jf, Magesan2011-hc, Mag12, Pro19, Pro22, Hin23, Hin23-2}. But integrating MCMs into this framework is challenging. RB protocols are designed to measure the average fidelity of \emph{invertible} quantum logic operations (i.e., gates), and most RB protocols execute \emph{motion reversal circuits}. A circuit (e.g., a sequence of $d$ Clifford operations) is followed by a circuit that implements its inverse.  The combined circuit's fidelity with the identity operation can be measured easily, and this is the core of standard RB. However, MCMs are inherently irreversible, collapsing the quantum state of the measured qubit(s). Motion-reversal circuits are not possible with MCMs. A new approach is required.

\customsection{Quantum instrument randomized benchmarking} QIRB is based on a new class of random Clifford circuits that contain MCMs (see Methods for full details). These circuits (Fig.~\ref{fig:qirb-combined}a-b) integrate MCMs into the \emph{$\Omega$-distributed random circuits} used in many modern RB protocols~\cite{proctor2021measuring, Hin23, Hin23-2, Pro19, Pro22, Polloreno2023-xa}. An $n$-qubit, depth-$d$ QIRB circuit consists of (i) $d$ circuit layers $L_1$, $\dots$, $L_d$ each sampled from a user-specified distribution $\Omega$, and (ii) additional ``dressing'' layers of randomizing single-qubit gates before and after each $\Omega$-distributed layer. Each layer $L_i$ specifies the parallel application of MCMs and Clifford gates to the $n$ qubits. QIRB only uses layers that (ideally) implement a computational basis measurement on $k \leq n$ qubits. Both ``reset-free'' (non-demolition) MCMs intended to leave a measured qubit in the state $\ket{k}$ conditional on observing ``$k$'', and ``reset'' MCMs intended to re-initialize each measured qubit to $\ket{0}$, can be used.

A $n$-qubit QIRB circuit $C$ containing $m$ MCMs produces $m+n$ readout bits.  We can use ``Pauli-tracking'' techniques introduced in binary RB \cite{Hin23-2} and first developed for direct fidelity estimation \cite{Moussa2012-rq, Fla11} to divide the $2^{m+n}$ possible output strings into two subsets of equal size called ``success'' and ``fail''.  Whenever a ``fail'' string is observed, it means that at least one error occurred in the circuit's execution. We quantify the error in a QIRB circuit $C$ with the metric $F = (N_{\textrm{success}} - N_{\textrm{fail}})/N$ \cite{Hin23-2}, where $N$ is the number of repetitions of $C$, and $N_{\textrm{success}}$ and $N_{\textrm{fail}}$ are the number of times a bit string from $C$'s ``success'' and ``fail'' sets (respectively) was observed.  We define $\overline{F}_d$ as the average value of $F$ over all depth-$d$ QIRB circuits.  The QIRB error rate ($\romega$) is the decay rate of $\overline{F}_d$ with respect to $d$.

The QIRB protocol is as follows: For a range of depths $d \geq 0$, sample $k_d$ $\Omega$-distributed QIRB circuits of depth $d$. Run each circuit $N \geq 1$ times to compute its $F$.  Average $F$ for all of the sampled depth-$d$ circuits to obtain an estimate of $\overline{F}_d$, and fit it to $\overline{F}_d = A(1 - \romega)^{d}$, where $A$ and $\romega$ are fit parameters. $A$ captures initial state preparation and final measurement error, and $\romega$ is the estimated \textit{QIRB error rate}.

Our theoretical analysis of QIRB (see Methods) shows that it satisfies all the expected properties of an RB protocol, and faithfully measures the rate at which MCMs cause errors in random circuits. We show that $\overline{F}_d$ decays exponentially in circuit depth $d$ (the key property of an RB protocol) and that the QIRB error rate $r_{\Omega}$ is bounded above and below by constant-factor multiples of the average layer infidelity ($\epsilon_{\Omega}$):  $\frac{3}{4}\epsilon_{\Omega} \leq \romega \leq  \frac{3}{2} \epsilon_{\Omega}$.  In most cases, $\romega \approx \epsilon_{\Omega}$.

\customsection{Trapped-ion experiments} We used QIRB to detect and quantify unexpected MCM crosstalk in Quantinuum's H1-1, a 20-qubit trapped-ion quantum computer. H1-1's qubits are encoded in the atomic hyperfine states of $\textsuperscript{171}\mathrm{Yb}^{+}$, with $\ket{0}$ corresponding to a ``dark'' state and $\ket{1}$ corresponding to a ``bright'' state that fluoresces when probed by a laser resonant with the $^2\mathrm{S}_{1/2}\ket{\mathrm{F}=1}\longleftarrow\, ^2\mathrm{P}_{1/2}\ket{\mathrm{F}=0}$ transition. Throughout the execution of a circuit, ions are stored in interleaved gate and auxiliary zones (Fig.~\ref{fig:qirb-combined}g). When an ion is measured, it (and another possibly unmeasured ion) is moved into a gate zone where the detection laser beam is pulsed on.  Stray light from the laser beam or the fluorescing ion can interact with the unmeasured ions to cause MCM-induced crosstalk errors.

Prior to this experiment, it was believed that the large physical separation between each gate zone and auxiliary zone would prevent MCM-induced crosstalk errors on the distant ions in the auxiliary zones, and that only the neighboring ion in a gate zone would be affected by stray light. As a result, Quantinuum's detailed system emulator accounted for local MCM-induced crosstalk errors (within each gate zone), but no long-range crosstalk. However, our QIRB experiments showed this assumption to be false, revealing significant MCM-induced crosstalk errors in the auxiliary zones.

We ran 2-qubit and 6-qubit QIRB experiments on the first two and six qubits (respectively) of H1-1, and we simulated the exact same circuits on Quantinuum's emulator.  Each QIRB experiment comprised 75 circuits, 15 each at depths $d = 0$, $1$, $4$, $32$, and $128$. The $\Omega$-distributed layers were sampled to contain a single MCM with probability $p_{\mcm}=0.35$ and a single $\cnot$ with probability $p_{\cnot}=0.2$.  Other experiments with different $p_\cnot$ and $p_\mcm$ showed consistent results (see Supplemental Note~\ref{app:h1-1-experiments}) and were used to estimate the MCM error rate (see Methods). We ran (and simulated) $N=100$ shots of each circuit to estimate $F$.  All experimentally estimated quantities are reported with $1\sigma$ error bars computed via a bootstrap.

The 2-qubit QIRB experiments yielded results ($\romega = 0.16\pm 0.02 \,\%$) consistent with emulator simulations ($\romega = 0.18\pm 0.02 \,\%$).  However, the 6-qubit experimental and emulated results showed a significant and unexpected discrepancy (Fig.~\ref{fig:qirb-combined}d). The emulator predicted $\romega = 0.17\pm 0.02 \,\%$ for 6-qubit QIRB, consistent with the 2-qubit results, while the $6$-qubit experiment revealed almost twice as much error ($\romega = 0.35\pm 0.03 \,\%$).  Because the QIRB error rate contains contributions from both gate errors and MCM errors, we used the analysis methods of Ref.~\cite{Hot23} to isolate the MCM error rate ($\epsilon_{\textrm{MCM}}$).  The discrepancy between 6-qubit experimental and emulated $\epsilon_{\textrm{MCM}}$ was even larger (Fig.~\ref{fig:qirb-combined}e).  Consistency between experimental and emulated 2-qubit QIRB results, and between emulated 2-qubit and 6-qubit QIRB results, indicated that the emulator was accurately modeling interactions between neighboring qubits in gate zones. This suggested that the significant discrepancy in the 6-qubit results was due to long-range crosstalk errors, which are not modeled in the emulator. 

We hypothesized that these errors were caused by stray light emanating from the detection beam used to measure ions or fluorescence from the measured ions in a circuit layer. To verify our hypothesis, we ran a bright-state depumping experiment~\cite{Gae21} to independently quantify the strength of MCM-induced crosstalk errors within the auxiliary and gate zones (see Supplemental Note~\ref{app:h1-1-experiments:ssec:bright-state}). These experiments confirmed our hypothesis (Fig.~\ref{fig:qirb-combined}f).

To definitively confirm long-range MCM-induced crosstalk as the cause, we implemented \textit{micromotion hiding}~\cite{Gae21} on distant ions during mid-circuit measurements, and then ran 6-qubit QIRB again to measure its efficacy.  Micromotion hiding reduces an unmeasured ion's sensitivity to scattered light by displacing the ion from the RF null, causing micromotion oscillations that Doppler-broaden the resonant detection light. This partially protects the unmeasured ion from stray photons emitted by the detection beam or a fluorescing measured ion during the measurement process. In the initial experiments, micromotion hiding was only applied within each gate zone, not to ions in the auxiliary zones.

When micromotion hiding was applied to ions in the auxiliary zone, the experimental 6-qubit QIRB $\romega$ dropped from from $0.35\pm 0.03 \,\%$ to $0.18\pm 0.03 \,\%$ (Fig.~\ref{fig:qirb-combined}d).  This value is consistent with the $0.17\pm 0.02 \,\%$ obtained from the emulator. The estimated $\varepsilon_\mcm$ dropped from $0.63\pm 0.09 \,\%$ to $0.29\pm 0.06 \,\%$ (Fig.~\ref{fig:qirb-combined}e), which is consistent with the $0.23\pm0.09 \,\%$ MCM error rate obtained from the emulator.  These 6-qubit experimental error rates are also consistent with those observed in 2-qubit experiments and simulations.  Finally, we confirmed the physical mechanism of the reduced error rates by repeating the bright-state depumping experiment with the extra micromotion hiding in the auxiliary zones. We observed a substantial drop in the unmeasured ions' sensitivity to scattered light (see Fig.~\ref{fig:qirb-combined}f).

\begin{figure}[!t]
    \centering
    \includegraphics[width=\linewidth]{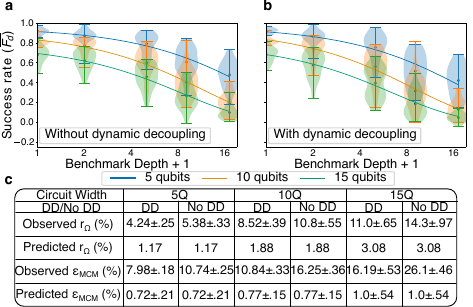}
\caption{\textbf{Quantifying the impact of dynamical decoupling on measurement-induced errors on \algiers.} \textbf{(a-b)} Success decay curves from performing 5-qubit (blue), 10-qubit (orange), and 15-qubit (green) QIRB on \algiers. Experiments were performed with (a) and without (b) dynamical decoupling (DD) on idling qubits during MCMs. \textbf{(c)} Even with DD, the observed QIRB error rate ($r_{\Omega}$) and the MCM error rate ($\epsilon_{\textrm{MCM}}$) are much larger than predicted by IBM's calibration data, suggesting significant MCM-induced crosstalk errors that are not mitigated by DD.}
    \label{fig:ibmq-exp-results}
\end{figure}



\customsection{Quantifying MCM crosstalk mitigation on \algiers} Next, we used QIRB to probe the nature of MCM errors, and measure the positive impact of dynamical decoupling, on an IBM Q processor (\algiers).  For this transmon-based processor, we had neither a high-fidelity emulator nor a detailed theoretical physics model for MCM errors, but its faster clock speed allowed more detailed experimental exploration.  We ran $5$-, $10$-, and $15$-qubit QIRB experiments on $\algiers$, with various values of $p_\cnot$ and $p_\mcm$. We focus here on experiments  with $p_\cnot = 0.3$ and $p_\mcm = 0.3$.  Other experiment designs are described in the Methods, and experimental results in Supplemental Note~\ref{app:ibmq-algiers-experiments}). 

To quantify the strength of crosstalk errors, we compared experimental QIRB results -- the estimated $\romega$ and error rates for one-qubit gates, two-qubit gates, and MCMs ($\varepsilon_{\mathrm{1Q}},$ $\varepsilon_{\mathrm{2Q}}$, and $\varepsilon_\mcm$, respectively) -- to predictions made using IBM's calibration data.  IBM's calibration data is gathered by running two-qubit RB, simultaneous one-qubit RB, and readout assignment error experiments.  These protocols are not directly sensitive to MCM crosstalk errors, so discrepancies between observed and predicted quantites can be ascribed to MCM crosstalk errors.

We observed significant gate- and MCM-induced crosstalk errors in $\algiers$ (Fig.~\ref{fig:ibmq-exp-results}a). The 5, 10, and 15-qubit QIRB experiments yield an experimental $r_{\Omega}$ that is $4-5\times$ larger than simulations using the calibration data (see Fig.~\ref{fig:ibmq-exp-results}c). Comparing observed and predicted $\varepsilon_{\textrm{MCM}}$ indicates that most of this discrepancy is due to measurement-induced crosstalk. The discrepancy between the observed and predicted $\varepsilon_{\textrm{MCM}}$ increases with $n$ (number of qubits), and for $n=15$, $\varepsilon_{\textrm{MCM}}$ is $26 \times$ larger than predicted.

MCM-induced errors in transmons can have multiple causes.  One mechanism is idling errors on unmeasured qubits.  MCMs have a longer duration than one- and two-qubit gates, so unmeasured qubits have more time to decohere and dephase. This can be mitigated by performing dynamical decoupling (DD)~\cite{viola1999dynamical} on the idling qubits, echoing away some of the idling errors. This approach is common and well-motivated, but its efficacy has not been systematically quantified \textit{in situ} (i.e., in the context of many-qubit circuits).

We ran QIRB circuits into which we introduced XX dynamical decoupling sequences~\cite{ibm-quantum-documentation}, and observed a significant reduction in the average layer error rate (Fig.~\ref{fig:ibmq-exp-results}c). The observed $\romega$ decreased by $\sim 22\%$ at each circuit width, and the estimated $\varepsilon_\mcm$ decreased by as much as $38\%$ (for $n=15$ qubits). This suggests that a large fraction of MCM-induced crosstalk errors are caused by idling errors on unmeasured qubits, and confirms that dynamical decoupling on idle qubits can reduce MCM-induced crosstalk significantly. However, the large residual discrepancy between predicted and observed error rates indicates that the majority of MCM-induced crosstalk errors are \textit{not} mitigatable in this way, and remain unexplained.

\customsection{Discussion} RB is ubiquitous precisely because it is generic.  RB protocols mix together all the errors in all of a system's gates to produce a single exponential decay and a single error rate that can be easily monitored, reported, and compared.  But until recently, motion reversal circuits---concatenating a random circuit with another circuit that inverts it---was the only way to achieve this simplicity.  QIRB builds on the recent innovations in binary RB (which sidestepped motion reversal) to extend the simplicity and genericity of RB to MCMs. As our deployment here illustrates, QIRB scales seamlessly to many qubits, and can easily measure error rates below $10^{-3}$.  It enables \textit{in situ} characterization of MCMs in processors designed and optimized to implement quantum error correction (QEC), paving the way for full, simultaneous characterization of errors in \textit{all} the key operations required for fault-tolerant QEC (one-qubit gates, two-qubit gates, and MCMs). We anticipate the techniques introduced here can be used to extend other RB protocols, especially those that measure leakage \cite{Wood2018-wf,Chen2016-sh, Chasseur2015-mw,Wallman2016-ne}, to work with MCMs.

\begin{small}
\section*{Methods}

\subsection{QIRB Circuits}\label{sec:protocol:ssec:mcmrb-circuits}
\begin{figure}
    \centering
    \includegraphics[width =\linewidth]{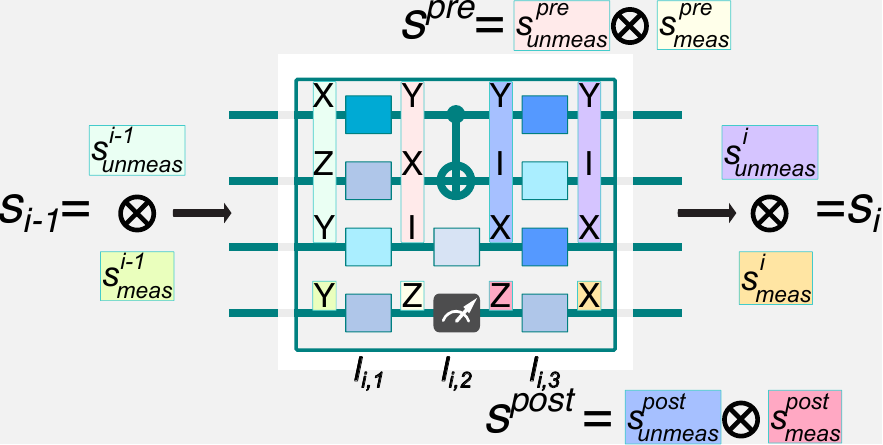}
    \caption{\textbf{A ``dressed'' QIRB circuit layer.} The bulk of a QIRB circuit consists of a sequence of ``dressed'' layers $\lbrace\tilde{L}_i\rbrace$, each composed of three sublayers: (i) $l_{i,1}$, (ii) $l_{i,2}$, and (iii) $l_{i,3}$. The middle sublayer, $l_{i,2}$ is an $\Omega$-distributed circuit layer, while the other two sublayers are specially crafted. Ideally, the state of the processor before the dressed layer $\tilde{L}_i$ is stabilized by the Pauli $s_{i-1}$. The first sublayer, $l_{i,1}$ is designed to rotate the state of the processor into an eigenstate of $\spre{\ }$, whose support, $\spre{\mathrm{meas}}$, on the measured qubits in $l_{i,2}$ is a tensor product of Pauli-$Z$ and $I$ operators. Likewise, $l_{i,3}$ is designed to (ideally) rotate the post-measurement state of the processor into an eigenstate of $s_i$, whose support on the measured qubits in $l_{i,2}$ is a uniform randomly chosen $m$-qubit Pauli operator, $s^{i}_{\mathrm{meas}}$.}
    \label{fig:qirb-layer}
\end{figure}
Conceptually, QIRB works by extending the Pauli tracking approach underlying both binary RB and direct fidelity estimation to circuits with mid-circuit measurements (see Supplemental Note~\ref{app:pauli-tracking}). An $n$-qubit QIRB circuit is designed so that, after each circuit layer, the $n$ qubits are (ideally) always in the $+1$ eigenspace of a particular $n$-qubit Pauli operator. To ensure this, each qubit that gets measured is immediately rotated to an eigenstate of a uniformly random single-qubit Pauli. Moreover, before each measurement, single-qubit gates are applied so that the measurement results can be compiled to determine if an error has occurred in the circuit. Thus a QIRB circuit $\tilde{C}$ mimics preparing the eigenstate of an $(n+m)$-qubit Pauli operator and performing an $(n+m)$-qubit Pauli measurement $s_{\tilde{C}}$. The success metric for $\tilde{C}$ is the expectation of its associated Pauli measurement, $\langle s_{\tilde{C}}\rangle$. 

The core component of a depth-$d$ QIRB circuit is a sequence of $d$ ``dressed'' $n$-qubit circuit layers, each composed of three sublayers, $\tilde{L} = l_3l_2l_1$. The middle layer $l_2$ is an $\Omega$-distributed layer and may contain MCMs.  We often refer to $l_2$ as $L$. Both $l_1$ and $l_3$ are layers of parallel single-qubit gates and/or idles. The single-qubit gates in $l_1$ are chosen to rotate the processor's state so that it is (ideally) an eigenstate of the $n$-qubit Pauli operator $\spre{} = \spre{\mathrm{meas}}\otimes \spre{\mathrm{unmeas}}$, where (i) $\spre{\mathrm{unmeas}}$ has support on the unmeasured qubits, (ii) $\spre{\mathrm{meas}}$ has support on the measured qubits, and (iii) $\spre{\mathrm{unmeas}}$ is the tensor product of Pauli-Z or identity operators (i.e., a $Z$-type Pauli) (see Fig.~\ref{fig:qirb-layer}). The gates in $l_3$ are selected so that the $k$ qubits measured in $l_2$ end up in the $+1$ tensor product eigenstate of a uniformly randomly sampled $k$-qubit Pauli operator (see Fig.~\ref{fig:qirb-layer}). When ``reset'' measurements are used, the $l_3$ layer is constructed to achieve this result assuming the measured qubits are in $\ket{0}^{\otimes k}$; when ``reset-free'' measurements are used, $l_3$ involves either $X_{\pi}$ gates conditioned on the measurement results \textit{or} a sign correction in post processing.

Formally, the construction of a depth-$d$ QIRB circuit $\tilde{C}$ is a sequential process that begins by randomly sampling a depth-$d$ $\Omega$-distributed circuit $C = L_d\cdots L_1$. Then, with $m$ being the number of MCMs in $C$, we sample a uniformly random $(n+m)$-qubit Pauli operator
\begin{equation}
    s = s(1) \otimes \cdots \otimes s(n+m)
\end{equation}
and a uniform random $+1$ tensor-product eigenstate of $s$
\begin{equation}
    \ket{\psi(s)} = \ket{\psi(s)_1}\otimes\cdots\otimes\ket{\psi(s)_{n+m}},
\end{equation}
where $\ket{\psi(s)_i}$ is stabilized by $s(i)$. The first layer $\tilde{L}_0$ in $\tilde{C}$ is a layer of single-qubit gates that locally prepares the first $n$-entries in $\ket{\psi(s)}$. Hence, in the ideal case, the state of the processor following $\tilde{L}_0$, denoted $\ket{\psi(s)}_{0}$ is stabilized by $s_0 \coloneqq s(1)\otimes\cdots\otimes s(n)$.

The $d$ dressed layers $\lbrace \tilde{L}_i\rbrace_{i=1}^{d}$ are constructed in an iterative fashion. In what follows, $U(l)$ denotes the unitary map induced by the circuit layer $l$ on any unmeasured qubits in $l$, and $k_i$ denotes the number of measured qubits in the $L_i$. The first dressed layer $\tilde{L}_1$ is constructed as follows:
\begin{enumerate}
    \item Set $l_{1,1}$ to be a layer of single-qubit gates such that $U(l_{1,1})s_0U(l_{1,1})^{-1} = \spre{0\mathrm{,meas}}\otimes\spre{0\mathrm{,unmeas}}$ is a $Z$-type Pauli on the $k_1$ measured qubits in $L_1$.
    \item Set $l_{1,2} = L_1$, the ``core'' $\Omega$-distributed circuit layer. Define $\spost{0} = U(l_{1,2})\spre{0,\mathrm{unmeas}}U(l_{1,2})^{-1} \otimes \spost{0,\mathrm{meas}}$, where $\spost{0,\mathrm{meas}} = \spre{0,\mathrm{meas}}$.
    \item Set $l_{1,3}$ to be a layer of single-qubit gates such that the measured qubits are re-prepared in a $+1$ tensor-product eigenstate of $s_1^{\mathrm{meas}} = s(n+1)\otimes\cdots\otimes s(n+k_1)$.
\end{enumerate}
In the ideal case, the state of the processor after $\tilde{L}_1$ should be stabilized by the Pauli, 
\begin{equation}
    s_1 \coloneqq U(\tilde{L}_1)\sunmeas{0}U(\tilde{L}_1)^{-1} \otimes s(n+1)\otimes\cdots s(n+k_1).
\end{equation}

The remaining dressed layers are constructed similarly using $s_{i-1}$ in place of $s_0$, and the final layer $\tilde{L}_{d+1}$ is constructed to transform $s_{d}$ into a $Z$-type Pauli operator, $s_{d+1}$. The above description leaves some freedom in the exact choice for each $l_{i,1}$ and $l_{i,3}$ and, in practice, we leverage this freedom to populate $l_{i,1}$ and $l_{i,3}$ with random single-qubit gates on any unmeasured qubits in order to implement randomized compilation for subsystem measurements~\cite{Bea23} (see Supplemental Note~\ref{app:subsystem-rc}). 

Overall, $\tilde{C}$ is constructed to mimic measuring the Pauli operator
\begin{equation}
    s_{\tilde{C}} \coloneqq \spre{0\mathrm{,meas}}\otimes\cdots\otimes s_{d+1},
\end{equation}
and we record the expectation of this Pauli operator $\langle s_{\tilde{C}}\rangle$. In practice, this is accomplished by concatenating the results from each MCM and each final computational basis state measurement into a bit string $b = (b_1,\ldots, b_{n+m})$, and then returning $1$ (``success'') if the corresponding computational basis state $\ket{b}$ is stabilized by $s_{\tilde{C}}$ and $-1$ (``failure'') otherwise.

\subsection{QIRB Theory}\label{sec:theory}
The QIRB protocol exhibits all of the expected properties of an RB protocol. Which is to say that, under reasonable assumptions about the noise, the average success rate of QIRB circuits ($\overline{F}_d$) decays exponentially in circuit depth, at a decay rate that is directly related to the rate at which MCMs and gates cause errors in random circuits. Moreover, this decay rate---or the QIRB error rate ($r_\Omega$)---is bounded above and below by constant-factor multiples of the average layer infidelity.

\subsubsection{Uniform stochastic instruments}
To analyze the QIRB protocol, we will model noisy circuit layers as uniform stochastic instruments~\cite{Bea23}. This is a specific ansatz generalizing the Pauli stochastic noise model to quantum instruments, and a processor's native layers may display errors that are not consistent with it -- but any layer $L$'s error can be transformed into uniform stochastic noise using randomized compiling (RC) for subsystem measurements~\cite{Bea23}.  RC for subsystem measurements can be implemented within QIRB (see Section~\ref{sec:protocol:ssec:mcmrb-circuits}); it is not required for the protocol, but is required to guarantee the validity of our theory describing it. As uniform stochastic instruments describe ``reset-free'' measurements, we introduce a refinement of the general uniform stochastic instrument model to describe ``reset'' measurements (see Supplemental Notes~\ref{app:theory:ssec:mcmrb+r} and~\ref{app:theory:ssec:mcmrb-r}).

A uniform stochastic instrument describes errors, in an $n$-qubit, $k$-measurement layer, that are the composition of a pre-measurement Pauli $X$ error $\Xa{a}$ on the measured qubits (specified by the $k$-bit string $a \in\mathbb{Z}_{2}^{k}$), with a post-measurement error that is the tensor product of a generic Pauli error $P$ on the unmeasured qubits and a bit-flip error $\Xa{b}$ on the measured qubits. Here, $\Xa{a}$ denotes the Pauli operator $\otimes_{i=1}^{k}X^{a_i}$.

We denote each possible error by its corresponding tuple $(a, P, b)$ and the probability of $(a, P, b)$ occurring in $L$ by $p_L(a,P,b)$. The process (entanglement) fidelity of a uniform stochastic instrument equals the probability of no error occurring~\cite{Mcl23}:
\begin{equation}\label{sec:preliminaries:eqn:layer-inifdelity}
    \mathcal{F}(L) = 1 - \sum_{(a,P,b)\neq (0,\mathbb{I}, 0)}{p_L(a,P,b)}.
\end{equation}
The $\Omega$-averaged layer infidelity is defined as
\begin{equation}\label{eq:avg_layer_infid}
    \eomega \coloneqq 1 - \mathbb{E}_{L\in\Omega}[\mathcal{F}(L)].
\end{equation}
In the absence of MCMs, $\eomega$ reduces to the same quantity that is (approximately) measured by most modern RB protocols including direct RB \cite{Pro19}, mirror RB \cite{Pro22}, binary RB \cite{Hin23-2}, and cross-entropy benchmarking \cite{boxio2018characterizing}.

\subsubsection{$\overline{F}_d$ decays exponentially}\label{ssec:theory:exp-decay}

To show that $\overline{F}_d$ decays exponentially, we need to calculate the probability $\mathrm{Pr}_{d}(S)$ of recording $+1$ (``success'') and the probability $\mathrm{Pr}_d(F)$ of recording $-1$ (``failure'') when executing a random depth-$d$ QIRB circuit. Ignoring end-of-circuit measurement error, we can calculate $\mathrm{Pr}_d(S)$ and $\mathrm{Pr}_d(F)$ recursively:
\begin{align} 
    \mathrm{Pr}_d(S) &= \mathrm{Pr}_{d-1}(S)(1-\ptrans{\ }) + \mathrm{Pr}_{d-1}(F)\ptrans{\ }\label{sec:theory:eqn:success-recursion} \\
    \mathrm{Pr}_d(F) &= \mathrm{Pr}_{d-1}(F)(1-\ptrans{\ }) + \mathrm{Pr}_{d-1}(S)\ptrans{\ }\label{sec:theory:eqn:failure-recursion},
\end{align}
where $\ptrans{\ }$ is the probability of errors in a random $\Omega$-distributed layer flipping a stabilizer state of a pre-layer Pauli into an anti-stabilizer state of the post-measurement Pauli. Equivalently, we can compute $\mathrm{Pr}_d(S)$ and $\mathrm{Pr}_d(F)$ by considering a two-state, discrete-time stochastic process with states $\lbrace S, F\rbrace$ and transition matrix
\begin{equation}
    M = \begin{pmatrix}
        1 - \ptrans{\ } & \ptrans{\ } \\
        \ptrans{\ } & 1 - \ptrans{\ }
        \end{pmatrix}.
\end{equation}
If $A$ is the probability of preparing in the $S$ (``success'') state and $B$ the probability of preparing in the $F$ (``failure'') state, then
\begin{align}
\overline{F}_d &\coloneqq \mathrm{Pr}_d(S) - \mathrm{Pr}_d(F) \label{eqn:f-bar-decay-1}\\
            &= (A-B)(1 - 2\ptrans{\ })^{d}. \label{eqn:f-bar-decay-2}
\end{align}
The equality of Eq.~\eqref{eqn:f-bar-decay-1} and Eq.~\eqref{eqn:f-bar-decay-2} follows from diagonalizing $M$ and noting that its eigenvalues are $1$ and $1-2\ptrans{\ }$.

So far we have ignored end-of-circuit measurement errors. These errors contribute a depth-independent co-factor as they modify the rate at which the S and F states are recorded in a depth-independent fashion. Hence,
\begin{equation}
    \overline{F}_d \propto (1-2\ptrans{\ })^d.
\end{equation}

This result shows that $\overline{F}_d$ decays exponentially and that $r_\Omega$ is well-defined, $r_\Omega = 2\ptrans{\ }$. Moreover, $\romega$ represents twice the rate at which errors flip stabilizer states to anti-stabilizer states in a circuit layer, a generalization of the interpretation of entanglement fidelity as twice the rate at which errors in a circuit layer anti-commute with a random Pauli. This interpretation allows use to compute $\romega$ for a given error model, and to relate $\romega$ to the average layer infidelity of that error model.

\subsubsection{Interpreting $r_\Omega$}\label{ssec:theory:romega}

Because $\romega$ is twice the rate at which errors flip stabilizer states to anti-stabilizers states in a circuit layer, it has the following form:
\begin{equation}\label{sec:theory:eqn:r-omega-sum}
    \romega = \E_{L\in\Omega}\Big[\sum_{a,P,b}\lambda_{L,a,P,b}\Big],
\end{equation}
where $\lambda_{L,a,P,b}$ is twice the rate at which the error $(a,P,b)$ occurring in a layer $L$ causes a state transition. We now argue that $\lambda_{L,a,P,b}$ depends upon: (i) the probability $p_L(a,P,b)$ at which $(a,P,b)$ occurs in $L$; (ii) the probability of $\Xa{a}$ anti-commuting with a random $\spre{\mathrm{meas}}$; (iii) the probability of $\Xa{b}$ anti-commuting with a random $\spost{\mathrm{meas}}$; and (iv) whether $P = \mathbb{I}$. Table~\ref{tab:mcmrb-contributions} summarizes an error $(a,P,b)$'s contribution to $\romega$.

\begin{table}[h]
\begin{tabular}{|c|c|}
\cline{1-2}
$P$  & Contribution to $r_\Omega$ \\\cline{1-2}
$\neq \mathbb{I}$ & $p_L(a, P, b)$ \\\cline{1-2}
$= \mathbb{I}$ & $2[\panti{a}+\panti{b}-2\panti{a}\panti{b}]p_L(a, P, b)$ \\\cline{1-2}
\end{tabular}
\caption{\textbf{Individual error contributions for QIRB.} Here we report the contribution to $\romega$ by the error $(a, P, b)$ occurring in layer $L$. Entries were calculated by determining the probability of each error causing a state transition. An error's contribution to $\romega$ depends upon three factors: (i) the probability $p_L(a,P,b)$ of the error occurring; (ii) if $P = \mathbb{I}$; and (iii) the probabilities $\panti{a}$ and $\panti{b}$ of $\Xa{a}$ and $\Xa{b}$ anti-commuting with a random $\spre{\mathrm{meas}}$ and $\spost{\mathrm{meas}}$, respectively. Because, on average, three-quarters of the entries of $\spre{\mathrm{meas}}$ and $\spost{\mathrm{meas}}$ are Pauli-$Z$ operators and one-quarter of their entries are $I$, computing $\panti{a}$ and $\panti{b}$ is a bit wonky. Supplemental Note~\ref{app:theory} contains more details on how to perform these calculations.}\label{tab:mcmrb-contributions}
\end{table}

An error $(a,P,b)$ causes a state transition when one (but not both) of the following occurs: (I) the pre-measurement error causes the reported state to be anti-stabilized by the pre-measurement Pauli instead of stabilized and (II) we prepare a post-measurement state that is anti-stabilized by the post-measurement Pauli instead of stabilized. Event (I) occurs when $\Xa{a}$ anti-commutes with $\spre{\mathrm{meas}}$. Event (II) occurs whenever $\Xa{b}$ anti-commutes with $\spost{\mathrm{meas}}$ (exclusive) or $P$ anti-commutes with $\spost{\mathrm{unmeas}}$. Averaging over all possible pre- and post-measurement Paulis, either event (I) or (II) occurs with probability 1/2 when $P = \mathbb{I}$. Otherwise, they occur with probability $\panti{a}+\panti{b}-2\panti{a}\panti{b}$. The values in Table~\ref{tab:mcmrb-contributions} come from multiplying these results by $2p_L(a,p,b)$.

\subsubsection{Bounding $\romega$}\label{ssec:bounding-romega}

We now prove the following bound on $\romega$,
\begin{equation}\label{eqn:mcmrb-romega-bound}
    \frac{3\eomega}{4} \leq \romega \leq \frac{3\eomega}{2}.
\end{equation}
This bound comes from analyzing the term 
\begin{equation*}
    \panti{a}+\panti{b}-2\panti{a}\panti{b}
\end{equation*}
in the $P=\mathbb{I}$ entry of Table~\ref{tab:mcmrb-contributions} for all possible values of $a$ and $b$. The possible values of this term form a discrete set on the two-dimensional surface in $\mathbb{R}^3$ defined by the equation $x+y-2xy$. The grid points are determined by the Hamming weight of $a$ and $b$. As shown in Supplemental Note~\ref{app:theory:ssec:bounding-e-omega}, this set has a minimum value of $3/8$ when either $a$ or $b$ have Hamming weight $1$, and a maximum value of $3/4$ whenever $a$ or $b$ have a Hamming weight of $2$, which completes the proof.

In practice, this bound is pessimistic. Closer analysis of $\panti{a}+\panti{b}-2\panti{a}\panti{b}$ shows that it rapidly approaches $1/2$ as the Hamming weights of $a$ and $b$ increase. Thus, most errors contribute roughly $p_L(a,P,b)/2$ to $\ptrans{\ }$ and $\romega \approx \eomega$.

\subsection{Simulations} \label{sec:simulations}
\begin{figure}
    \centering
    \includegraphics[width =\linewidth]{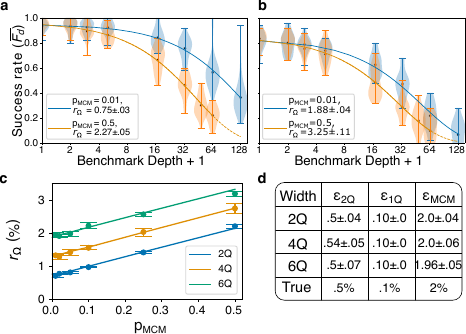}
    \caption{\textbf{Simulation of QIRB on few-qubit QPUs with a simple error model.} \textbf{(a-b)} RB curves from two two-qubit [(a)] and two [(b)] six-qubit QIRB-r simulations performed using $p_{\cnot} = .35$. The points show $\overline{F}_{d}$, and the violin plots show the distributions of $\langle s_{\tilde{C}}\rangle$ over circuits of that depth. QIRB has succeeded in all four cases as evidenced by the exponential decay of $\overline{F}_d$. \textbf{(c)} A plot comparing $r_\Omega$ from, respectively, eight two-, four-, and six-qubit simulations with varying $p_{\mcm}$ and fixed $p_{\cnot} = .35$ (the points) against our predictions for $r_\Omega$ (the lines) based on the data generating model. The error bars are $1\sigma$ error bars generated by performing each simulation eight times. These results validate QIRB as the $r_{\Omega}$ observed in each simulation matches that predicted by our theory. \textbf{(d)} Table containing the extracted one-qubit gate, $\cnot$, and MCM error rates (resp. $\varepsilon_{\mathrm{1Q}}\text{, }\varepsilon_{\mathrm{2Q}}\text{, and } \varepsilon_{\mcm}$). We see close agreement between the estimated error rates and those used to generate the model. See Supplemental Note~\ref{app:simulations:ssec:error-rate-extraction} for additional details.}
    \label{fig:sim-results}
\end{figure}

We used simulations to validate the predictions of our QIRB theory: $\overline{F}_d$ decays exponentially at a rate $r_\Omega$ predicted by our theory. To do so, we simulated QIRB on simulated noisy 2, 4, and 6-qubit quantum computers with native single-qubit Cliffords, $\cnot$ gates with all-to-all connectivity, and MCMs with post-measurement reset. We modelled gate errors as local depolarizing channels and noisy measurements as perfect computational basis measurements preceded by a bitflip error channel. We used the same layer sampling rules as in our trapped-ion and IBMQ experiments, with $p_\cnot$ values of $0.2$, $0.35$, $0.5$ and $p_\mcm$ values of $0.01$, $0.02$, $0.05$, $0.10$, $0.25$, $0.5$. See Supplemental Notes~\ref{app:simulations:ssec:simulation-details} for additional details on the noise model, experiment designs, and full results.

Fig.~\ref{fig:sim-results} summarizes the results from our simulations. The data shows close agreement with our theory: the decays are exponential, $r_{\Omega}$ is well predicted by our theory, and our extracted error rates agree with those used to generate the data, These results confirm QIRB's validity in simulation. 

\subsection{Trapped-ion experimental details}\label{sec:quantinuum-methods}

We conducted QIRB experiments on Quantinuum's H1-1, utilizing the first 2 and 6 qubits of the device, both with and without micromotion hiding (MMH) in the auxiliary zones. Using the layer rules described in the main text, we constructed QIRB experiment designs using $p_\cnot$ equal to $0.2$ and $0.35$, and $p_\mcm$ of 0.2 and 0.3. For each experiment, we generated a total of 15 circuits at the depths $d = 0$, $1$, $4$, $32$, and $128$ resulting in 75 circuits per experiment, with each circuit executed 100 times. We also simulated each experiment on Quantinuum's emulator. We employed the fitting procedure outlined in Supplemental Note~\ref{app:simulations:ssec:error-rate-extraction} to extract estimated one-qubit gate, two-qubit gate, and MCM error rates from both the experiments and simulations. Additionally, we generated 100 bootstrapped samples to assess the uncertainty of and 30 bootstrapped samples to estimate the uncertainty in the extracted error rates. Results from each experiment are detailed in Table~\ref{tab:h1-1-rb-numbers} of Supplemental Note~\ref{app:h1-1-experiments}.

\subsection{\algiers\ demonstrations details}\label{sec:ibmq-algiers-methods}
We conducted QIRB demonstrations on the first 5, 10, and 15 qubits in \algiers. The layer rules applied in our \algiers\ demonstrations were consistent with those used in our trapped-ion experiments, with the exception that we limited the placement of $\cnot$
gates to the edges of \algiers' connectivity graph. We used $p_\cnot$ values of $0.2$, $0.35$, and $0.5$ and $p_\mcm$ values of $0.2$, $0.3$, and $0.4$. Results from all of these experiments were used to estimate one-qubit, two-qubit, and MCM error rates using an error rates model (see Supplemental Note~\ref{app:simulations:ssec:error-rate-extraction}). The QIRB error rates for each demonstration can be found in Table~\ref{tab:ibmq-algiers-rb-numbers}, along with bootstrapped standard deviations derived from 30 bootstrapped samples.

\section*{Data and code availability}
The code used to implement our protocol is available within \texttt{pyGSTi}\cite{nielsen2022pygstio}. Data will be made available upon reasonable request.

\section*{Acknowledgements} 
This material was funded in part by the U.S. Department of Energy, Office of Science, Office of Advanced Scientific Computing Research, Quantum Testbed Pathfinder Program. This research was funded in part by the Office of the Director of National Intelligence (ODNI), Intelligence Advanced Research Projects Activity (IARPA). Sandia National Laboratories is a multi-program laboratory managed and operated by National Technology and Engineering Solutions of Sandia, LLC., a wholly owned subsidiary of Honeywell International, Inc., for the U.S. Department of Energy's National Nuclear Security Administration under contract DE-NA-0003525. This research used IBM Quantum resources of the Air Force Research Laboratory. All statements of fact, opinion or conclusions contained herein are those of the authors and should not be construed as representing the official views or policies of the U.S. Department of Energy, IARPA, the ODNI, the U.S. Government, IBM, or the IBM Quantum team.

\end{small}

\bibliography{bibliography}

\appendix
\onecolumngrid
\newpage
\section*{SUPPLEMENTARY INFORMATION}

\section{Pauli tracking and binary RB}\label{app:pauli-tracking}
Pauli tracking is a procedure for estimating the fidelity of a quantum circuit $C$, whose noisy superoperator we denote by $\phi(C) = \mathcal{E}\circ\mathcal{U}(C)$, where $\mathcal{U}(C)$ is the superoperator representation of the unitary $U(C)$ and $\mathcal{E}$ is an error channel. For Clifford circuits, Pauli tracking reduces to estimating
\begin{equation}
    \mathbb{E}_{s\in\mathbb{P}_n}\Tr(s'\phi(C)[s]),
\end{equation}
where $s' \coloneqq U(C)sU(C)^{\dagger}$. In this framework, $s$ is the initial Pauli operator that is tracked through the layers of $C$, and $s'$ is the target Pauli operator that $s$ is transformed into by the circuit. Because Pauli operators are not valid quantum states, in practice Pauli tracking is implemented by inputting random eigenstates of a random Pauli operator $s$ into $C$ and then measuring the target Pauli operator $s'$ at the end of the circuit. 

For layer sets containing only Clifford layers, the fidelity of a depth-$d$ $\Omega$-distributed circuit $C = L_d\cdots L_1$ may be estimated using an ensemble of depth-$d$ binary RB (BiRB) circuits. For a fixed Pauli $s$, $\Tr(s\phi(C)[s]))$ is estimated by generating many depth-$d$, $n$-qubit BiRB circuits $\tilde{C}$ as follows: (i) sample a uniformly random state $\ket{\psi(s)}$ from the set of tensor product eigenstates of $s$; (ii) construct a layer $L_0$ of single-qubit gates that prepares $\ket{\psi(s)}$; (iii) prefix the core circuit $C$ with $L_0$; (iv) postfix a layer $L_{d+1}$ of single-qubit gates that transforms $s'$ into a $Z$-type Pauli (i.e., the tensor product of single-qubit $Z$ Pauli and identity I Pauli operators). Each BiRB circuit has a target Pauli $s_{\tilde{C}} \coloneqq U(L_{d+1})s'U(L_{d+1})^{-1}$, and we have
\begin{equation}
    \Tr(s'\phi(C)[s]) = \mathbb{E}_{\tilde{C}}[\langle s_{\tilde{C}}\rangle].
\end{equation}
BiRB estimates the average fidelity of depth-$d$ $\Omega$-distributed circuits by sampling many depth-$d$ $\Omega$-distributed circuits $\mathcal{C} = \lbrace C\rbrace$, constructing ensembles of BiRB circuits for each $C \in \mathcal{C}$, and then computing the sample average of each circuit's success metric $\langle s_{\tilde{C}}\rangle$ across the BiRB circuits.   

When circuits contain MCMs, the natural generalization of Pauli tracking is to track larger Pauli operators. For a depth-$d$, $n$-qubit, $m$-measurement $\Omega$-distributed circuit $C$, we track $(n+m)$-qubit Pauli operators. However, tracking a $(n+m)$-qubit Pauli with only $n$-qubits available presents two challenges: how do we use an $n$-qubit circuit to mimic preparing the eigenstate of a $(n+m)$-qubit Pauli $s$ and how do we use it to propagate $s$ through the circuit to determine the overall $(n+m)$-qubit Pauli measurement? The core insight behind the development of QIRB circuits is that we can use the mid-circuit measurements to introduce virtual qubits to address these challenges.

\section{QIRB performs subsystem measurement randomized compilation}\label{app:subsystem-rc}

In this supplemental note, we show that our QIRB circuits twirl errors on the mixed layers into uniform stochastic error channels, assuming that the single-qubit gate error is small.

We will show that the random single-qubit gate layers $l_1$ and $l_3$ in the composite layers of QIRB circuits, combined with the QIRB data analysis for mid-circuit measurement results, effectively perform randomized compiling on the composite layers. It is sufficient to show that these operations satisfy the three requirements for randomized compiling of subsystem measurements ~\cite{Bea23}:
\begin{enumerate}
\item Random Pauli gates are applied on all unmeasured qubits.
\item On each measured qubit, $\Za{a}$ is performed before the measurement, and $\Za{b}$ is performed after the measurement, for uniform random $a,b \in \mathbb{Z}_2$. 
\item On each measured qubit, $\Xa{-x}$ is performed for uniform random $x \in \mathbb{Z}_2$, and $\Xa{x}$ is performed after the measurement. Add $x$ to the classical result of the measurement. 
\end{enumerate}

We start by showing that the required quantum operations are applied by considering a single composite layer,
$\mathbb{E}_{l_3}\mathbb{E}_{l_2}\mathbb{E}_{l_1} \Lambda(l_3) \Lambda(l_2) \Lambda(l_1)$. We represent each imperfect single-qubit gate layer with a layer-independent pre-layer error channel, i.e., $\Lambda(l) = \mathcal{U}(l) \mathcal{E}$, and we represent each mixed layer as a quantum instrument $\Lambda(l_2) = \sum_{x \in \mathbb{Z}_2^{n_{\textrm{mcm}}(l_2)}} {M_x(l_2)} \otimes \mathbf{e_x}$, where $n_{n_{\textrm{mcm}}(l_2)}$ denotes the number of qubits measured in $l_2$ and $\mathbf{e_x}$ denotes the classical outcome encoded by bit string $x$.

\begin{align}
    \mathbb{E}_{l_3}\mathbb{E}_{l_2}\mathbb{E}_{l_1} \Lambda(l_3) \Lambda(l_2) \Lambda(l_1) &= \mathbb{E}_{l_3}\mathbb{E}_{l_2}\mathbb{E}_{l_1} \mathcal{U}(l_3)\mathcal{E} \mathcal{E}_{l_2}\mathcal{U}(l_2) \mathcal{U}(l_1) \mathcal{E}\\
    & = \mathbb{E}_{l'_3}\mathbb{E}_{l_2}\mathbb{E}_{l'_1} \mathcal{U}(l'_3)\mathcal{P}^c\mathcal{E} \left(\sum_{x \in \mathbb{Z}_2^{n_{\textrm{mcm}}(l_2)}} {M_x(l_2)} \otimes \mathbf{e_{x+k(\mathcal{P})}} \right)\mathcal{P}\mathcal{U}(l'_1) \mathcal{E}, 
\end{align}
for any $n$-qubit Pauli operators $\mathcal{P},\mathcal{P}^c$. Here, we define $k(\mathcal{P})$ to be the bit string encoding the corrections to the MCM results accounting for $\mathcal{P}$, i.e., $k(\mathcal{P})_i=1$ if $\mathcal{P}_i=X$ or $Y$ and $k(\mathcal{P})=0$ otherwise. We can sample $\mathcal{P}$ uniformly from the $n$-qubit Pauli operators and pick $\mathcal{P}^c$ uniformly from the Pauli operators satisfying (a) $\mathcal{P}^c_i=P_i$ if qubit $i$ is not measured in $l_2$ (b) $\mathcal{P}^c_i = X$ or $Y$ if qubit $i$ is measured in $l_2$ and $\mathcal{P}_i=X$ or $Y$, and $\mathcal{P}'_i = Z$ or $I$ otherwise. These sampling distributions cover the conditions in (1) and (2) above. Furthermore, the QIRB data analysis satisfies condition (3), i.e., it implements the transformation from $\mathbf{e_x}$ to $\mathbf{e_{x+k(\mathcal{P})}}$ which informally says that the target measurement result is changed if randomized compiling induces a bit flip. If $\mathcal{P}_i=X$ or $Y$, and the stabilizer on the measurement is $s_i = Z$, then the stabilizer transforms to $-s_i$. If $s_i = I$, the measurement result is discarded (and thus can be classically manipulated without consequence). 
Therefore, by averaging over the Pauli operators, the error channel for dressed layer in an QIRB circuit is  a uniform stochastic instrument, 
\begin{align}
    \mathbb{E}_{l_3}\mathbb{E}_{l_2}\mathbb{E}_{l_1} \Lambda(l_3) \Lambda(l_2) \Lambda(l_1)  = \mathbb{E}_{l'_3}\mathbb{E}_{l_2}\mathbb{E}_{l'_1} \mathcal{U}(l'_3) \overline{\mathcal{M}}(l_2)\mathcal{U}(l'_1) \mathcal{E}, 
\end{align}
where
\begin{equation}
\overline{\mathcal{M}}(l_2) = \sum_{a,b,j \in \mathbb{Z}_2^m} \mathcal{U}(l_2)\mathcal{T}_{a,b} \otimes \sket{j+b}\sbra{j+a} \otimes \mathbf{e_j},
\end{equation}
where $\mathcal{U}(l_2)$ denotes the unitary superoperator of the error-free operation on the unmeasured qubits in $l_2$, $\mathbf{e_j}$ denotes the classical outcome, and $\sket{j+a} = \sket{\ket{j+a}\bra{j+a}}$ is the vector in Hilbert-Schmidt space representing the computational basis state $\ket{j+a}$ specified by the bitstring $j+a$. 

\section{Theory of QIRB}\label{app:theory}
In this supplemental note, we provide a (mostly) self-contained presentation of how to compute the QIRB error rate $\romega$ when using measurements with and without post-measurement resets (resp. QIRB-r and QIRB+r). We begin by reiterating our assumptions, which we then use to derive explicit formulations of noise models used in our QIRB-r and QIRB+r theory. After establishing explicit noise models, we show how to calculate $\romega$ for both QIRB+r (Supplemental Note~\ref{app:theory:ssec:mcmrb+r}) and QIRB-r (Supplemental Note~\ref{app:theory:ssec:mcmrb-r}) in terms of individual error probabilities. We conclude with a proof of a bound on $\romega$,
\begin{equation*}
    \frac{3\eomega}{4} \leq \romega \leq \frac{3\eomega}{2}.
\end{equation*}
\subsection{Assumptions and noise models}\label{app:theory:ssec:assumptions-and-noise}
We make the following three assumptions: 
\begin{enumerate}
    \item\label{ass:uniform} the pre-layer Pauli $s_i$ is uniformly random throughout a circuit $C$,
    \item\label{ass:markovian} the noise is Markovian,
    \item\label{ass:uniform-stochastic} and that each noisy layer $L$ is described by a uniform stochastic instrument,
    \begin{equation}\label{app:eqn:uniform-stochastic-instruments}
        \varphi(L) = \sum_{a,b,j\in\mathbb{Z}_{2}^{m}}{\mathcal{U}(L)\Tab{a}{b} \otimes \sket{j+b}\sbra{j+a}\otimes \mathbf{e_j}},
    \end{equation}
\end{enumerate}

Uniform stochastic instruments are nice because they only model three kinds of errors~\cite{Bea23}:
\begin{enumerate}
    \item readout error, i.e., the system was in the state $\ket{j+a}$, but is reported as $\mathbf{e_j}$,
    \item wrong basis state errors, i.e., the system is left in the state $\ket{j+b}$ when the outcome $\mathbf{e_j}$ is read,
    \item and, a stochastic error is applied to the unmeasured qubits which is independent of the reported outcome. 
\end{enumerate}
Despite the simplicity of uniform stochastic instruments, the form used in equation~\ref{app:eqn:uniform-stochastic-instruments} is challenging to work with. Instead, we want to re-interpret a uniform stochastic instrument as the probabilistic sum of maps with the following decomposition: (i) a pre-measurement X-type Pauli error $\Xa{a}$ on $\Hmeas$; (ii) a post-measurement X-type Pauli error $\Xa{b}$ on $\Hmeas$; and (iii) an arbitrary Pauli error $P$ on $\Hunmeas$. Alternatively, if $L$ is an $n$-qubit, $m$-measurement layer:
\begin{equation}\label{app:theory:eqn:alternate-formulation}
    \varphi(L) = \sum_{P\in\mathbb{P}_{n-m}}\sum_{a,b\in\mathbb{Z}_2^{m}}{p(a,P,b)\sum_{j\in\mathbb{Z}^m_2}[\mathcal{U}(L)\circ\mathcal{U}(P)]\otimes[\mathcal{U}(\Xa{b})\circ\sket{j}\sbra{j}\circ\mathcal{U}(\Xa{a})]\otimes \mathbf{e_j}}.
\end{equation}
We model each layer in QIRB-r as in equation~\ref{app:theory:eqn:alternate-formulation}. We slightly modify the noise model in QIRB+r to account for the post-measurement resets:
\begin{equation}\label{app:theory:eqn:reset-error-model}
    \varphi(L) = \sum_{P\in\mathbb{P}_{n-m}}\sum_{a,b\in\mathbb{Z}_2^{m}}{p_{L}(a,P,b)\sum_{j\in\mathbb{Z}^m_2}[\mathcal{U}(L)\circ\mathcal{U}(P)]\otimes[\mathcal{U}(\Xa{b})\circ\sket{0}\sbra{j}\circ\mathcal{U}(\Xa{a})]\otimes \mathbf{e_j}}.
\end{equation}

\subsection{Calculating $\romega$ for QIRB+r}\label{app:theory:ssec:mcmrb+r}
We now show how to compute an individual error's contribution to the $\Omega$-average layer error rate measured by QIRB+r.
Recall that $\romega$ has the following form (equation~\ref{sec:theory:eqn:r-omega-sum}):
\begin{equation}
    \romega = \E_{L\in\Omega}\Big[\sum_{a,P,b}\lambda_{L,a,P,b}\Big].
\end{equation}
Thus, an error $(a,P,b)$ in an $n$-qubit, $m$-measurement layer $L$'s individual contribution to $\romega$ is quantified by $\lambda_{L,a,P,b}$, or how sensitive $\romega$ is to the error. We will show that 
\begin{equation}\label{app:theory:eqn:mcmrb+r-lambda}
    \lambda_{L,a,P,b} =     
    \begin{cases}
        2[\panti{a}-2\panti{a}\panti{b}+\panti{b}]p_L(a,P,b), & P = \mathbb{I},\\
        p_L(a,P,b), & P \neq \mathbb{I},
    \end{cases}
\end{equation}
where $\panti{a}$ and $\panti{b}$ are the probabilities of $\mathbb{I}\otimes\Xa{a}$ and $\mathbb{I}\otimes\Xa{b}$ anti-commuting with a pre-measurement and post-measurement Pauli, respectively.

As described in Section~\ref{sec:theory}, $\romega$ measures twice the rate at which errors flip stabilizer states to anti-stabilizer states (and vice-a-versa) in an QIRB circuit. Therefore, $\lambda_{L,a,P,b}$ is twice the rate at which the error $(a,P,b)$ causes such a transition multiplied by the probability of the error having occurred (conditioned on the probability of selecting the circuit layer). These transitions occur precisely when an odd number of the following criteria are met:
\begin{enumerate}
    \item\label{app:theory:cond:mcmrb+r-one} The pre-measurement error $\mathbb{I}\otimes\Xa{a}$ causes the readout to report a measurement bitstring corresponding to a computational basis state that is anti-stabilized by $\spre{\ }$ instead of stabilized (or vice-a-versa).
    \item\label{app:theory:cond:mcmrb+r-two} The post-measurement error $\mathbb{I}\otimes\Xa{b}$ leads to the preparation of a post-measurement state that is anti-stabilized by $\spost{\ }$.
    \item\label{app:theory:cond:mcmrb+r-three} The error on the unmeasured qubits $P\otimes\mathbb{I}$ anti-commutes with $\spost{\ }$.
\end{enumerate}
Conditions~\ref{app:theory:cond:mcmrb+r-one} and~\ref{app:theory:cond:mcmrb+r-two} are equivalent to $\mathbb{I}\otimes\Xa{a}$ anti-commuting with $\spre{\ }$ and $\mathbb{I}\otimes\Xa{b}$ anti-commuting with $\spost{\ }$, respectively. Hence, we can compute $\lambda_{L,a,P,b}$ by determining how often an odd number of the above conditions are met. As each criteria is independent of the other, we will calculate probability each event occurring independently.

Both condition~\ref{app:theory:cond:mcmrb+r-one} and condition~\ref{app:theory:cond:mcmrb+r-two} occur at the same rate. Without loss of generality, we consider condition~\ref{app:theory:cond:mcmrb+r-one}. Treating $\spre{\ }$ as a random variable, condition~\ref{app:theory:cond:mcmrb+r-one} is met whenever $\Xa{a}$ anti-commutes with an $m$-qubit Pauli $s$ drawn from the distribution
\begin{equation}\label{app:theory:eqn:pauli-distribution}
    \mu(s) =     
    \begin{cases}
        0, & \text{if } s \text{ is not Z-type},\\
        \frac{3^{wt(s)}}{4^m}, & \text{else,}
    \end{cases}.
\end{equation}
where $wt(s)$ is the number of non-identity tensor elements in $s$. We denote this probability by $\panti{a}$. It is equal to 
\begin{equation}\label{app:theory:eqn:panti}
   \panti{a} = \sum\limits_{\substack{1\leq i\leq wt(a)\textrm{,} \\ i \textrm{ is odd}}}{{wt(a) \choose i} \Big(\frac{3}{4}\Big)^i\Big(\frac{1}{4}\Big)^{wt(a)-i}}
\end{equation}
Analogously, condition~\ref{app:theory:cond:mcmrb+r-two} is met with probability $\panti{b}$.

Condition~\ref{app:theory:cond:mcmrb+r-three} is met at different rates depending upon if $P = \mathbb{I}$ or not. When $P = \mathbb{I}$, it anti-commutes with $\spost{\ }$ with probability 0. When $P \neq \mathbb{I}$, condition~\ref{app:theory:cond:mcmrb+r-three} is met with probability $1/2$.

Having determined the probabilities of conditions~\ref{app:theory:cond:mcmrb+r-one},~\ref{app:theory:cond:mcmrb+r-two}, and~\ref{app:theory:cond:mcmrb+r-three}, we may easily calculate the probability of an odd number of conditions being met. We specialize to the case where $P \neq \mathbb{I}$. A single condition is met with probability
\begin{equation}
    \frac{\panti{a}[1-\panti{b}]}{2} + \frac{[1-\panti{a}][1-\panti{b}]}{2} + \frac{[1-\panti{a}]\panti{b}}{2} = \frac{1-\panti{a}\panti{b}}{2}.
\end{equation}
Whereas all three conditions are met with probability
\begin{equation}
    \frac{\panti{a}\panti{b}}{2}.
\end{equation}
In total, an odd number of conditions are met with probability
\begin{equation}
    \frac{1-\panti{a}\panti{b}}{2} + \frac{\panti{a}\panti{b}}{2} = \frac{1}{2}.
\end{equation}
Thus $\lambda_{L,a,P,b} = 2*1/2*p_L(a,P,b)$ when $P \neq \mathbb{I}$. The analogous calculation establishes the $P = \mathbb{I}$ case. 

This concludes our proof of an individual error's contribution to $\romega$ in QIRB+r. The error's contribution depends upon the layer $L$ that it occurs in, how often the error occurs in the layer, and $\lambda_{L,a,P,b}$, or QIRB+r's sensitivity to the error occuring in layer $L$. 

\subsection{Calculating $\romega$ for QIRB-r}\label{app:theory:ssec:mcmrb-r}
QIRB-r enjoys the same sensitivity to errors as QIRB+r, although the reasoning is slightly different. As before, $\romega$ measures the average rate at which errors flip stabilizer states to anti-stabilizer states (and vice-a-versa) in an QIRB circuit. The only difference in our analysis is that we must consider the classically conditioned X gates included in an QIRB-r circuit. Therefore condition~\ref{app:theory:cond:mcmrb+r-two} becomes
\begin{enumerate}
    \item[($2^\prime$)]\label{app:theory:cond:mcmrb-r-two-prime} The classically controlled post-measurement $X$ gates and the post-measurement error $\mathbb{I}\otimes\Xa{b}$ lead to the preparation of a state that is anti-stabilized by $\spost{\ }$.
\end{enumerate}
If $\mathsf{e_j}$ is the measurement result and $\spost{\mathrm{meas}} = \Za{d}$ on the measured qubits, then condition~\ref{app:theory:cond:mcmrb-r-two-prime} is equivalent to \begin{equation}\label{eqn:bad-corr-factor}
    d\cdot (j+b) \not\equiv d\cdot j \mod 2 \Leftrightarrow d\cdot b \not\equiv 0 \mod 2.
\end{equation}
Of course, this is again equivalent to $\mathbb{I}\otimes\Xa{b}$ anti-commuting with $\spost{\ }$. Hence, QIRB-r enjoys the same sensitivity to errors as QIRB+r. 

\subsection{Bounding $\romega$}\label{app:theory:ssec:bounding-e-omega}
We now give a proof of the bounds on $\romega$ for QIRB (equation~\ref{eqn:mcmrb-romega-bound}),
\begin{equation*}
    \frac{3\eomega}{4} \leq \romega \leq \frac{3\eomega}{2},
\end{equation*}
in terms of the $\Omega$-distributed average layer infidelity $\eomega$. Establishing these bounds reduces to analyzing $\romega$'s sensitivity to errors $(a, P, b)$ where $P = \mathbb{I}$ (i.e., the non-constant terms in Tables~\ref{tab:mcmrb-contributions}). In fact, equation~\ref{eqn:mcmrb-romega-bound} follows immediately from the following inequality
\begin{equation}\label{app:theory:eqn:mcmrb-factor}
    \frac{3}{8} \leq \panti{a}-2\panti{a}\panti{b}+\panti{b} \leq \frac{3}{4},
\end{equation}
which we prove in the following lemma. 

\begin{lemma}[Equation~\ref{app:theory:eqn:mcmrb-factor}]\label{app:theory:lemma:monomial-bound}
On the domain $D = \lbrace (a, b) \mid a, b\in\mathbb{Z}_{2}^m\rbrace\setminus\lbrace (0,0)\rbrace$,
    \begin{equation}\label{app:theory:eqn:mcmrb+r-factor}
        \frac{3}{8} \leq \panti{a}-2\panti{a}\panti{b}+\panti{b} \leq \frac{3}{4}.
    \end{equation}
\end{lemma}
\begin{proof}
    It is easier to establish the inequalities in equation~\ref{app:theory:eqn:mcmrb-factor} by establishing the same bounds for the function $f(x,y) = x-2xy-y$ on the real-value domain $\tilde{D} \coloneqq (0,\infty)\times (0,\infty) \setminus (0,1)\times (0,1)$. In this case, the second partial derivative test establishes that $f(x,y)$ obtains its extreme values on the boundary of its domain:
    \begin{equation*}
        \partial\tilde{D} = \mathrm{Ray}[(1,0)\rightarrow (\infty,0)]\cup \mathrm{Ray}[(0,1)\rightarrow (0,\infty)]\cup \mathrm{LineSeg}[(1,0)\rightarrow (1,1)] \cup \mathrm{LineSeg}[(0,1)\rightarrow (1,1)],
    \end{equation*}
    where $\mathrm{Ray}$ denotes the ray between two points and $\mathrm{LineSeg}$ denotes the line segment between two points. Again, the second derivative test shows that $f(x,y)$ obtains its extreme values on the boundaries of its boundary segments. Put differently, $f(x,y)$ obtains its extrema at the points $(1,0), (0,1),$ and/or $(1,1)$. Inspection shows that $f(x,y)$ obtains a maximum of $f(1,1) = 3/4$ and a minimum of $f(1,0) = f(0,1) = 3/8$. 
\end{proof}

\section{QIRB simulations}\label{app:simulations}
\subsection{Simulations}\label{app:simulations:ssec:simulation-details}

\begin{table*}[h!]
\centering
\begin{tabular}{|c|l|l|}
\cline{1-3}
Circuit width & Two-qubit gate and MCM density & $r_\Omega$ (\%) \\\cline{1-3}
\multirow{18}{*}{2Q}& $p_\cnot = 0.2$ and $p_\mcm = 0.01$ & $0.672\pm0.056$ \\\cline{2-3}
& $p_\cnot = 0.2$ and $p_\mcm = 0.02$ & $0.711\pm0.047$ \\\cline{2-3}
& $p_\cnot = 0.2$ and $p_\mcm = 0.05$ & $0.800\pm0.045$ \\\cline{2-3}
& $p_\cnot = 0.2$ and $p_\mcm = 0.10$ & $0.924\pm0.046$ \\\cline{2-3}
& $p_\cnot = 0.2$ and $p_\mcm = 0.25$ & $1.381\pm0.060$ \\\cline{2-3}
& $p_\cnot = 0.2$ and $p_\mcm = 0.50$ & $2.139\pm0.083$ \\\cline{2-3}
& $p_\cnot = 0.35$ and $p_\mcm = 0.01$ & $0.722\pm0.030$ \\\cline{2-3}
& $p_\cnot = 0.35$ and $p_\mcm = 0.02$ & $0.744\pm0.034$ \\\cline{2-3}
& $p_\cnot = 0.35$ and $p_\mcm = 0.05$ & $0.820\pm0.034$ \\\cline{2-3}
& $p_\cnot = 0.35$ and $p_\mcm = 0.10$ & $0.979\pm0.033$ \\\cline{2-3}
& $p_\cnot = 0.35$ and $p_\mcm = 0.25$ & $1.429\pm0.067$ \\\cline{2-3}
& $p_\cnot = 0.35$ and $p_\mcm = 0.50$ & $2.201\pm0.048$ \\\cline{2-3}
& $p_\cnot = 0.5$ and $p_\mcm = 0.01$ & $0.773\pm0.013$  \\\cline{2-3}
& $p_\cnot = 0.5$ and $p_\mcm = 0.02$ & $0.769\pm0.045$  \\\cline{2-3}
& $p_\cnot = 0.5$ and $p_\mcm = 0.05$ & $0.903\pm0.033$  \\\cline{2-3}
& $p_\cnot = 0.5$ and $p_\mcm = 0.10$ & $1.031\pm0.034$ \\\cline{2-3}
& $p_\cnot = 0.5$ and $p_\mcm = 0.25$ & $1.481\pm0.055$ \\\cline{2-3}
& $p_\cnot = 0.5$ and $p_\mcm = 0.50$ & $2.187\pm0.048$ \\\cline{1-3}
\multirow{18}{*}{4Q}& $p_\cnot = 0.2$ and $p_\mcm = 0.01$ & $1.284\pm0.033$ \\\cline{2-3}
& $p_\cnot = 0.2$ and $p_\mcm = 0.02$ & $1.275\pm0.040$ \\\cline{2-3}
& $p_\cnot = 0.2$ and $p_\mcm = 0.05$ & $1.356\pm0.068$ \\\cline{2-3}
& $p_\cnot = 0.2$ and $p_\mcm = 0.10$ & $1.521\pm0.050$ \\\cline{2-3}
& $p_\cnot = 0.2$ and $p_\mcm = 0.25$ & $1.935\pm0.080$ \\\cline{2-3}
& $p_\cnot = 0.2$ and $p_\mcm = 0.50$ & $2.648\pm0.065$ \\\cline{2-3}
& $p_\cnot = 0.35$ and $p_\mcm = 0.01$ & $1.307\pm0.039$ \\\cline{2-3}
& $p_\cnot = 0.35$ and $p_\mcm = 0.02$ & $1.319\pm0.073$ \\\cline{2-3}
& $p_\cnot = 0.35$ and $p_\mcm = 0.05$ & $1.405\pm0.072$ \\\cline{2-3}
& $p_\cnot = 0.35$ and $p_\mcm = 0.10$ & $1.576\pm0.060$ \\\cline{2-3}
& $p_\cnot = 0.35$ and $p_\mcm = 0.25$ & $2.030\pm0.094$ \\\cline{2-3}
& $p_\cnot = 0.35$ and $p_\mcm = 0.50$ & $2.726\pm0.110$ \\\cline{2-3}
& $p_\cnot = 0.5$ and $p_\mcm = 0.01$ & $1.342\pm0.034$ \\\cline{2-3}
& $p_\cnot = 0.5$ and $p_\mcm = 0.02$ & $1.416\pm0.055$ \\\cline{2-3}
& $p_\cnot = 0.5$ and $p_\mcm = 0.05$ & $1.479\pm0.066$ \\\cline{2-3}
& $p_\cnot = 0.5$ and $p_\mcm = 0.10$ & $1.622\pm0.051$ \\\cline{2-3}
& $p_\cnot = 0.5$ and $p_\mcm = 0.25$ & $2.039\pm0.056$ \\\cline{2-3}
& $p_\cnot = 0.5$ and $p_\mcm = 0.50$ & $2.759\pm0.084$ \\\cline{1-3}
\multirow{18}{*}{6Q}& $p_\cnot = 0.2$ and $p_\mcm = 0.01$ & $1.872\pm0.081$ \\\cline{2-3}
& $p_\cnot = 0.2$ and $p_\mcm = 0.02$ & $1.873\pm0.102$ \\\cline{2-3}
& $p_\cnot = 0.2$ and $p_\mcm = 0.05$ & $1.964\pm0.060$ \\\cline{2-3}
& $p_\cnot = 0.2$ and $p_\mcm = 0.10$ & $2.174\pm0.055$ \\\cline{2-3}
& $p_\cnot = 0.2$ and $p_\mcm = 0.25$ & $2.475\pm0.078$ \\\cline{2-3}
& $p_\cnot = 0.2$ and $p_\mcm = 0.50$ & $3.256\pm0.215$ \\\cline{2-3}
& $p_\cnot = 0.35$ and $p_\mcm = 0.01$ & $1.872\pm0.090$ \\\cline{2-3}
& $p_\cnot = 0.35$ and $p_\mcm = 0.02$ & $1.912\pm0.063$ \\\cline{2-3}
& $p_\cnot = 0.35$ and $p_\mcm = 0.05$ & $2.046\pm0.074$ \\\cline{2-3}
& $p_\cnot = 0.35$ and $p_\mcm = 0.10$ & $2.183\pm0.102$ \\\cline{2-3}
& $p_\cnot = 0.35$ and $p_\mcm = 0.25$ & $2.559\pm0.097$ \\\cline{2-3}
& $p_\cnot = 0.35$ and $p_\mcm = 0.50$ & $3.260\pm0.119$ \\\cline{2-3}
& $p_\cnot = 0.5$ and $p_\mcm = 0.01$ & $1.969\pm0.073$ \\\cline{2-3}
& $p_\cnot = 0.5$ and $p_\mcm = 0.02$ & $1.932\pm0.078$ \\\cline{2-3}
& $p_\cnot = 0.5$ and $p_\mcm = 0.05$ & $2.059\pm0.074$ \\\cline{2-3}
& $p_\cnot = 0.5$ and $p_\mcm = 0.10$ & $2.202\pm0.072$ \\\cline{2-3}
& $p_\cnot = 0.5$ and $p_\mcm = 0.25$ & $2.615\pm0.107$ \\\cline{2-3}
& $p_\cnot = 0.5$ and $p_\mcm = 0.50$ & $3.321\pm0.109$ \\\cline{1-3}
\end{tabular}
\caption{\textbf{QIRB simulation results.} QIRB error rates $r_{\Omega}$ and $1\sigma$ error bars from our LoQS simulations. The numbers were generated by running 8 different simulations for a fixed set of hyperparameters (e.g., sampling parameters, circuit depths).}\label{app:simulations:tab:results}
\end{table*}

We performed simulations of QIRB on noisy $n = 2, 4, \textrm{and } 6$ qubit QPUs using pyGSTi~\cite{Nie20} to generate the QIRB circuits and the Logical Qubit Simulator (LoQS)~\cite{LoQS} to run the simulations. We used a native gate set of all 24 single-qubit Clifford gates and two-qubit $\cnot$ gates. Measurements included post-measurement resets.

Gate errors were modelled as gate-independent, local depolarizing channels preceding the gates and supported on the same qubit acted upon by the gate. MCM and final measurement errors were modelled as bitflip error channels preceding the measurement. We used a bitflip rate of $2\%$, and one and two-qubit gate fidelities of $99.9\%$ and $99.5\%$, respectively. 

Table~\ref{app:simulations:tab:results} lists the QIRB error rate $r_\Omega$ generated by each simulation. We also report the average $r_\Omega$ across all eight simulations for a fixed distribution $\Omega$, and the sample standard deviation of $r_\Omega$.

\subsection{Predicting $r_\Omega$ for a gate-independent, local depolarizing, and bitflip readout error model}\label{app:simulations:ssec:eomega-depolarizing}

In this section, we derive an explicit expression for the the QIRB error rate $r_\Omega$ under the gate-independent, local depolarizing, measurement bitflip error model used in our simulations. In this error model,
\begin{itemize}
\item each single-qubit gate has single-qubit depolarizing noise, with the probability of $X, Y,$ and $Z$ errors being $p_X=p_Y=p_Z=\varepsilon_1$,
\item each two-qubit gate has two-qubit depolarizing noise, with the probability of each Pauli error being $\varepsilon_2$,
\item and each (single-qubit) MCM has single-qubit bit flip error with error probability $p_X = \varepsilon_m$.
\end{itemize}
Hence single-qubit gates have a fidelity of $F_{1Q} = 1 - 3\varepsilon_1$ and two-qubit gates have a fidelity of $F_{2Q} = 1 - 15\varepsilon_2$.

Let $\ptrans{L}$ denote the rate at which errors in layer $L$ cause a state transition. Because $r_\Omega$ is the average of $2\ptrans{L}$ over $L\in\Omega(L)$, we reduce the calculation to determining $\ptrans{L}$ for an $n$-qubit layer $L$ with $k_m$ mid-circuit measurements, $k_1$ single-qubit gates, and $k_2$ two-qubit gates. Moreover, due to the tensor product structure of the noise, $\ptrans{L}$ is simply the sum of two independent events: (i) the probability of witnessing an error that anti-commutes with $\spre{\mathrm{meas}}$ and commutes with $\spre{\mathrm{unmeas}}$ and (ii) the probability of witnessing an error that commutes with $\spre{\mathrm{meas}}$ and anti-commutes with $\spre{\mathrm{unmeas}}$.

The probability of an error on the unmeasured qubits anti-commuting with a random Pauli is
\begin{equation}
    \panti{\mathrm{unmeas}} \coloneqq \frac{1 - F_{1Q}^{k_1}F_{2Q}^{k_2}}{2}.
\end{equation}

To calculate the rate at which an error on the measured qubits anti-commutes with $\spre{\mathrm{meas}}$, we average over $\mu(s)$ (equation~\ref{app:theory:eqn:pauli-distribution}). We consider a weight $w$ bitflip error. As we can have up to weight $k_m$ errors, the rate at which errors on the measured qubits anti-commute with $\spre{\mathrm{meas}}$ is
\begin{equation}\label{app:simulations:eqn:panti-mcm}
    \panti{\mathrm{meas}, k_m} \coloneqq \sum_{w=0}^{k_m}{{k_m \choose w}\epsilon_m^w(1-\epsilon_m)^{k_m-w}\panti{w}}.
\end{equation}

Altogether, we have
\begin{equation}\label{app:simulations:eqn:ptrans-ugly}
    \panti{L} = \panti{\mathrm{unmeas}}[1-\panti{\mathrm{meas}, k}] + [1-\panti{\mathrm{unmeas}}]\panti{\mathrm{meas}, k}
\end{equation}
For clarity, we define the \emph{effective fidelity} of the measured subsystem to be
\begin{equation}\label{app:simulations:eqn:mcm-fidelity}
    F^{k_m}_{\mcm} \coloneqq 1 - 2\panti{\mathrm{meas}, k_m}.
\end{equation}
Then, we can re-write equation~\ref{app:simulations:eqn:ptrans-ugly} as
\begin{equation}\label{app:simulations:eqn:ptrans-nice}
    \ptrans{L} = \frac{1 - F_{1Q}^{k_1}F_{2Q}^{k_2}F_{\mcm}^{k_m}}{2}.
\end{equation}

\subsection{Error rate extraction}\label{app:simulations:ssec:error-rate-extraction}
We calculated the estimated error rates reported in the paper by fitting error rates models (ERMs)~\cite{Hot23} to the data. ERMs are a class of parameterized models used to model a circuit's overall error rate. In an ERM, each circuit component is assigned an error rate, the set of which is denoted $\mathcal{E}$. An ERM predicts a circuit's overall error rate by counting how often each circuit component appears in the circuit, and then using those counts as inputs to a $\mathcal{E}$-parametrized function.

We used ERMs that predict $\langle s_C\rangle$ by effectively assuming the gate-independent, local depolarizing, bitflip readout error model described in~\ref{app:simulations:ssec:eomega-depolarizing}. Each ERM  $E$ is a four-parameter model, consisting of an single-qubit gate error rate $\varepsilon_{1Q}$, a two-qubit gate error rate $\varepsilon_{2Q}$, an MCM error rate $\varepsilon_{\mcm}$, and a generic SPAM error rate $\varepsilon_{\mathrm{SPAM}}$. The ERM's prediction of $\langle s_C\rangle$ for a depth-$d$ circuit $C$ is given by:
\begin{equation}
    E(C) = \varepsilon_{\mathrm{SPAM}}\prod_{i=0}^{d+1}{(1-\varepsilon_{1Q})^{k_{1,i}}(1-\varepsilon_{2Q})^{k_{2,i}}F_{\mcm}^{k_{m,i}}(\varepsilon_{\mcm})},
\end{equation}
where $k_{1,i}$ and $k_{2,i}$ are the number of one- and two-qubit gates in $L_i$, respectively, and $F_{\mcm}^{k_{m,i}}(\varepsilon_{\mcm})$ is the effective fidelity of layer $L_i$'s measured subsystem assuming a bitflip readout error of $\varepsilon_{\mcm}$.

Each ERM was fit using the data from several QIRB simulations or experiments, which differed only in the $p_{\cnot}$ and $p_{\mcm}$ used. For instance, when fitting an ERM to the two-qubit simulations in Section~\ref{sec:simulations}, we fit the ERM to the data in all of the two-qubit QIRB simulations listed in Supplemental Note~\ref{app:simulations:ssec:simulation-details}. We chose to fit each ERM to multiple simulations or experiments in order to improve the stability of the fit. 

We fit each ERM by minimizing the mean squared error of an ERM's prediction for $\langle s_C\rangle$ with its actual value using scipy's built-in Nelder-Mead optimizer

\section{\algiers\ demonstrations}\label{app:ibmq-algiers-experiments}

\begin{table*}[h!]
\begin{tabular}{|l|l|l|l|}
\cline{1-4}
Circuit width & Two-qubit gate and MCM density & $r_\Omega$ (with XXDD) & $r_\Omega$ (w/o XXDD) \\\cline{1-4}
5Q & $p_\cnot = .2$ and $p_{\mcm} = .2$ & $3.31\pm.19$ & $4.43\pm.32$ \\\cline{1-4}
5Q & $p_\cnot = .2$ and $p_{\mcm} = .3$ & $4.59\pm.23$ & $5.97\pm.34$ \\\cline{1-4}
5Q & $p_\cnot = .2$ and $p_{\mcm} = .4$ & $5.64\pm.24$ & $7.51\pm.34$ \\\cline{1-4}
5Q & $p_\cnot = .3$ and $p_{\mcm} = .2$ & $3.33\pm.22$ & $4.40\pm.36$ \\\cline{1-4}
5Q & $p_\cnot = .3$ and $p_{\mcm} = .3$ & $4.24\pm.25$ & $5.38\pm.33$ \\\cline{1-4}
5Q & $p_\cnot = .3$ and $p_{\mcm} = .4$ & $6.30\pm.25$ & $7.89\pm.34$ \\\cline{1-4}
5Q & $p_\cnot = .4$ and $p_{\mcm} = .2$ & $3.58\pm.28$ & $5.09\pm.32$ \\\cline{1-4}
5Q & $p_\cnot = .4$ and $p_{\mcm} = .3$ & $5.65\pm.25$ & $7.45\pm.43$ \\\cline{1-4}
5Q & $p_\cnot = .4$ and $p_{\mcm} = .4$ & $7.62\pm.39$ & $9.92\pm.53$ \\\cline{1-4}
10Q & $p_\cnot = .2$ and $p_{\mcm} = .2$ & $5.85\pm.28$ & $7.84\pm.55$ \\\cline{1-4}
10Q & $p_\cnot = .2$ and $p_{\mcm} = .3$ & $7.81\pm.50$ & $10.21\pm.55$ \\\cline{1-4}
10Q & $p_\cnot = .2$ and $p_{\mcm} = .4$ & $8.78\pm.38$ & $12.3\pm.81$ \\\cline{1-4}
10Q & $p_\cnot = .3$ and $p_{\mcm} = .2$ & $6.23\pm.34$ & $8.40\pm.59$ \\\cline{1-4}
10Q & $p_\cnot = .3$ and $p_{\mcm} = .3$ & $8.52\pm.39$ & $10.81\pm.55$ \\\cline{1-4}
10Q & $p_\cnot = .3$ and $p_{\mcm} = .4$ & $9.73\pm.43$ & $13.0\pm.65$ \\\cline{1-4}
10Q & $p_\cnot = .4$ and $p_{\mcm} = .2$ & $6.41\pm.26$ & $8.46\pm.38$ \\\cline{1-4}
10Q & $p_\cnot = .4$ and $p_{\mcm} = .3$ & $8.02\pm.46$ & $11.10\pm.71$ \\\cline{1-4}
10Q & $p_\cnot = .4$ and $p_{\mcm} = .4$ & $9.96\pm.41$ & $13.80\pm.88$ \\\cline{1-4}
15Q & $p_\cnot = .2$ and $p_{\mcm} = .2$ & $9.42\pm.47$ & $12.65\pm.85$\\\cline{1-4}
15Q & $p_\cnot = .2$ and $p_{\mcm} = .3$ & $11.92\pm.44$ & $16.13\pm.91$ \\\cline{1-4}
15Q & $p_\cnot = .2$ and $p_{\mcm} = .4$ & $13.86\pm.71$ & $20.38\pm1.66$ \\\cline{1-4}
15Q & $p_\cnot = .3$ and $p_{\mcm} = .2$ & $10.14\pm.57$ & $13.12\pm.71$ \\\cline{1-4}
15Q & $p_\cnot = .3$ and $p_{\mcm} = .3$ & $11.04\pm.65$ & $14.31\pm.97$ \\\cline{1-4}
15Q & $p_\cnot = .3$ and $p_{\mcm} = .4$ & $13.26\pm.69$ & $18.44\pm1.26$ \\\cline{1-4}
15Q & $p_\cnot = .4$ and $p_{\mcm} = .2$ & $10.20\pm.69$ & $13.86\pm1.00$ \\\cline{1-4}
15Q & $p_\cnot = .4$ and $p_{\mcm} = .3$ & $13.54\pm.63$ & $19.35\pm1.24$ \\\cline{1-4}
15Q & $p_\cnot = .4$ and $p_{\mcm} = .4$ & $14.42\pm.73$ & $21.81\pm1.41$ \\\cline{1-4}
\end{tabular}
\caption{\textbf{\algiers\ QIRB error rates.} QIRB error rates $r_{\Omega}$ from our \algiers\ demonstrations. The error rates are reported as percents.}\label{tab:ibmq-algiers-rb-numbers}
\end{table*}

\begin{table*}[h!]
\begin{tabular}{|l|l|l|l|l|l|l|}
\cline{1-7}
\text{Circuit width} & $\varepsilon_{\mathrm{1Q}}$ (w/ XXDD) & $\varepsilon_{\mathrm{1Q}}$ (w/o XXDD) & $\varepsilon_{\mathrm{2Q}}$ (w/ XXDD) & $\varepsilon_{\mathrm{2Q}}$ (w/o XXDD) & $\varepsilon_{\mcm}$ (w/ XXDD) & $\varepsilon_{\mcm}$ (w/o XXDD) \\\cline{1-7}
5Q & $0.11\pm.03$ & $0.13\pm.04$ & $2.21\pm.34$ & $2.77\pm.37$ & $7.98\pm.18$ & $10.74\pm.25$ \\\cline{1-7}
10Q & $0.26\pm.03$ & $0.29\pm.03$ & $2.40\pm.55$ & $3.01\pm.61$ & $10.84\pm.33$ & $16.25\pm.36$ \\\cline{1-7}
15Q & $0.24\pm.04$ & $0.35\pm.07$ & $5.27\pm1.63$ & $1.52\pm3.10$ & $16.19\pm.53$ & $26.14\pm.46$ \\\cline{1-7}
\end{tabular}
\caption{\textbf{\algiers\ estimated error rates.} Estimated one-qubit gate ($\varepsilon_{\mathrm{1Q}}$), two-qubit gate ($\varepsilon_{\mathrm{2Q}}$), and MCM error rates ($\varepsilon_{\mcm}$) from our \algiers\ demonstrations with and without XXDD. All errors are reported as percents.}\label{tab:ibmq-algiers-errors}
\end{table*} 

We performed demonstrations of QIRB on the first 5, 10, and 15 qubits in \algiers. We used the same layer rules for our \algiers\ demonstrations as in our simulations and trapped-ion experiments, except we restricted the placement of $\cnot$ gates to edges in \algiers' connectivity graph and used $p_{\cnot}$ equal to $.2$, $.35$, and $.5$ along with $p_{\mcm}$ equal to $.2$, $.3$, and $.4$. Our \algiers\ demonstrations consisted of 30 circuits at each of the following depths: 0, 1, 4, 8, and 16; performed 1000 times each. The QIRB error rate for each demonstration is located in Table~\ref{tab:ibmq-algiers-rb-numbers}. Bootstrapped standard deviations (obtained from generating 30 bootstrapped samples) are reported as well.

Table~\ref{tab:ibmq-algiers-errors} contains the average one-qubit gate, two-qubit gate, and MCM error rates extracted from our IBM demonstrations. These error rates were extracted using the the process outlined in Supplemental Note~\ref{app:simulations:ssec:error-rate-extraction}, and the standard deviation were obtained from 30 bootstrapped samples. 

We also used \algiers\ calibration data (see Supplemental Note~\ref{app:ibmq-algiers-calibration-data}) to predict $r_{\Omega}$ and $\varepsilon_{\mcm}$ for our demonstrations. The predicted $\varepsilon_{\mcm}$ is the average of the reported readout error of the involved qubits. Predictions of $r_{\Omega}$ we made by averaging the gate errors on the involved qubits and edges, and then making a prediction using the local depolarizing model described in Supplemental Note~\ref{app:simulations:ssec:eomega-depolarizing}.

\section{H1-1 experiments}\label{app:h1-1-experiments}
\subsection{QIRB experiments}\label{app:h1-1-experiments:ssec:mcmrb-experiments}
We ran QIRB experiments on Quantinuum's H1-1 using the first 2 and 6 qubits in the device with and without micromotion hiding (MMH) in the device's auxiliary zones. We used the same layer rules for our H1-1 experiments as in our simulations, except we only used $p_{\cnot} = .2$ and $.35$ and $p_{\mcm}$ of $.2$ and $.3$. For each experiment, we generated a total of 15 circuits per depth at depths $d = 0$, $1$, $4$, $32$, and $128$ for a total of 75 circuits per experiment. We ran each circuit 100 times. We used the fitting procedure described in Supplemental Note~\ref{app:simulations:ssec:error-rate-extraction} to extract estimated two-qubit and MCM error rates from both the experiments and simulations. We also generated 100 bootstrapped samples to estimate the uncertainty of $\romega$ and 30 bootstrapped samples to estimate the uncertainty in the extracted error rates. 

We also ran the same circuit sets on Quantinuum's in-house H1-1 emulator to generate predicted $r_{\Omega}, \varepsilon_{\mathrm{2Q}},$ and $\varepsilon_{\mcm}$ as a baseline. Statistically significant discrepancies in the actual and baseline values suggest the presence of unmodeled noise. In order to estimate the uncertainty of our baseline values, we ran an additional 5 simulations per combination of $d$, $p_\cnot$, and $p_\mcm$. Overall, our estimated one-qubit and two-qubit gate errors are in line with Quantinuum's reported values~\cite{productsheet}, with the exception of the two-qubit gate error rate extracted from the experiments with micromotion hiding. These low values are likely due to the optimizer finding a local minimum near zero for several bootstrapped samples. 

\begin{table}[h]
  \centering
  \begin{tabular}{|c|c|c|c|c|c|c|c|c|c|}
    \hline
    \multirow{2}{*}{Circuit width} & \multicolumn{3}{c|}{One-qubit error ($\epsilon_{\mathrm{1Q}}$)} & \multicolumn{3}{c|}{Two-qubit error ($\epsilon_{\mathrm{2Q}}$)} & \multicolumn{3}{c|}{MCM error ($\epsilon_{\mcm}$)} \\
    \cline{2-10}
     & MMH & No MMH & Emulator & MMH & No MMH & Emulator & MMH & No MMH & Emulator \\
    \hline
    2 & $.04\pm.02$ & $.04\pm.02$ & $.02\pm.01.$ & $.01\pm.01$ & $.07\pm.05$ & $.16\pm.07$ & $.26\pm.06$ & $.25\pm.06$ & $.28\pm.06$\\
    \hline
    6 & $.01\pm.005$ & $.004\pm.006$ & $.01\pm.005$ & $.06\pm.06$ & $.38\pm.12$ & $.22\pm.08$ & $.29\pm.06$ & $.63\pm.09$ & $.23\pm.09$ \\
    \hline
  \end{tabular}
  \caption{\textbf{H1-1 estimated error rates.} Estimated one-qubit gate ($\varepsilon_{\mathrm{1Q}}$), two-qubit gate error rates ($\varepsilon_{\mathrm{2Q}}$), and MCM error rates ($\varepsilon_{\mcm}$) from our H1-1 experiments and H1-1 emulator simulations. All errors are reported as percents.}\label{tab:h1-1-errors}
\end{table}

\begin{table*}[h]
\begin{tabular}{|l|l|l|l|l|}
\cline{1-5}
Circuit width & Two-qubit gate and MCM density & $r_\Omega$ (w/ MMH) & $r_\Omega$ (w/o MMH) & $r_\Omega$ (emulator) \\\cline{1-5}
2Q & $p_\cnot = .20$ and $p_{\mcm} = .2$ & $.15\pm.01$ & $.13\pm.02$ & $.13\pm.01$ \\\cline{1-5}
2Q & $p_\cnot = .20$ and $p_{\mcm} = .3$ & $.17\pm.01$ & $.18\pm.02$ & $.19\pm.02$ \\\cline{1-5}
2Q & $p_\cnot = .35$ and $p_{\mcm} = .2$ & $.12\pm.01$ & $.16\pm.02$ & $.18\pm.02$ \\\cline{1-5}
2Q & $p_\cnot = .35$ and $p_{\mcm} = .3$ & $.16\pm.01$ & $.18\pm.02$ & $.21\pm.02$ \\\cline{1-5}
6Q & $p_\cnot = .20$ and $p_{\mcm} = .2$ & $.18\pm.02$ & $.27\pm.04$ & $.17\pm.01$ \\\cline{1-5}
6Q & $p_\cnot = .20$ and $p_{\mcm} = .3$ & $.22\pm.02$ & $.38\pm.02$ & $.18\pm.02$ \\\cline{1-5}
6Q & $p_\cnot = .35$ and $p_{\mcm} = .2$ & $.18\pm.03$ & $.35\pm.03$ & $.17\pm.02$ \\\cline{1-5}
6Q & $p_\cnot = .35$ and $p_{\mcm} = .3$ & $.22\pm.02$ & $.45\pm.04$ & $.25\pm.03$ \\\cline{1-5}
\end{tabular}
\caption{\textbf{QIRB error rates.} The $r_{\Omega}$ from our H1-1 experiments and H1-1 emulator simulations. All RB numbers are reported as percents. Only the 6-qubit, $p_\cnot = .35$ and $p_\mcm = .2$ results were reported in the main text.}\label{tab:h1-1-rb-numbers}
\end{table*}

\subsection{Bright-state depumping experiment}\label{app:h1-1-experiments:ssec:bright-state}

\begin{table}[h!]
    \centering
    \begin{tabular}{|c|*{5}{c|}}
        \hline
        & Qubit 0 & Qubit 1 & Qubit 2 & Qubit 3 & Qubit 4 \\
        \hline
        Without MMH & $0.0017 \pm 0.0004$ & $0.0016 \pm 0.0004$ & $0.0021 \pm 0.0005$ & $0.0016 \pm 0.0003$ & $0.0011 \pm 0.0003$ \\
        \hline
        With MMH & $0.0017\pm.0004$ & $0.0018 \pm 0.0004$ & $0.0020 \pm 0.0004$ & $0.0031 \pm 0.0006$ & $0.0014 \pm 0.0003$ \\
        \hline
    \end{tabular}
    \caption{\textbf{Bright-state depumping experimental results.} Observed process infidelity of the bright-state depumping process on the qubits in H1-1's gate zones. Results are reported as percents with $1\sigma$ error bars estimated from a parametric bootstrap.}
    \label{app:h1-1-experiments:tab:bright-state-results}
\end{table}

\begin{table}[h!]
    \centering
    \begin{tabular}{|c|*{5}{c|}}
        \hline
        & Qubit 5 & Qubit 6 & Qubit 7 & Qubit 8 & Qubit 9 \\
        \hline
        Without MMH & $0.0015 \pm 0.0004$ & $0.0022 \pm 0.0005$ & $0.0014 \pm 0.0005$& $0.0060 \pm 0.0006$ & $0.0016 \pm 0.0003$ \\
        \hline
        With MMH & $0.0019 \pm 0.0005$ & $0.0028 \pm 0.0005$ & $0.0009 \pm 0.0003$ & $0.0051 \pm 0.0007$ & $0.0013 \pm 0.0004$\\
        \hline
    \end{tabular}
    \caption{\textbf{Bright-state depumping experimental results cont.} Observed process infidelity of the bright-state depumping process on the qubits in H1-1's gate zones. Results are reported as percents with $1\sigma$ error bars estimated from a parametric bootstrap.}
    \label{app:h1-1-experiments:tab:bright-state-results-cont}
\end{table}

\begin{table}[h!]
    \centering
    \begin{tabular}{|c|*{5}{c|}}
        \hline
        & Qubit 10 & Qubit 11 & Qubit 12 & Qubit 13 & Qubit 14 \\
        \hline
        Without MMH & $0.0142 \pm 0.0008$ & $0.0148 \pm 0.0008$
 & $0.0515 \pm 0.0016$ & $0.1450 \pm 0.0044$ & $0.0544 \pm 0.0018$\\
        \hline
        With MMH & $0.0008 \pm 0.0002$ & $0.0008 \pm 0.0002$ & $0.0022 \pm 0.0003$ & $0.0034 \pm 0.0004$ & $0.0037 \pm 0.0004$ \\
        \hline
    \end{tabular}
    \caption{\textbf{Bright-state depumping experimental results cont.} Observed process infidelity of the bright-state depumping process on the qubits in H1-1's auxiliary zones. Results are reported as percents with $1\sigma$ error bars estimated from a parametric bootstrap.}
    \label{app:h1-1-experiments:tab:bright-state-results-cont-2}
\end{table}

\begin{table}[h!]
    \centering
    \begin{tabular}{|c|*{5}{c|}}
        \hline
        & Qubit 15 & Qubit 16 & Qubit 17 & Qubit 18 & Qubit 19 \\
        \hline
        Without MMH & $0.0498 \pm 0.0016$ & $0.0406 \pm 0.0012$ & $0.0400 \pm 0.0012$ & $0.0402 \pm 0.0015$ & $0.1144 \pm 0.0033$ \\
        \hline
        With MMH & $0.0023 \pm 0.0003$ & $0.0018 \pm 0.0003$ & $0.0020 \pm 0.0003$ &  $0.0034 \pm 0.0004$ & $0.0040 \pm 0.0004$\\
        \hline
    \end{tabular}
    \caption{\textbf{Bright-state depumping experimental results cont.} Observed process infidelity of the bright-state depumping process on the qubits in H1-1's auxiliary zones. Results are reported as percents with $1\sigma$ error bars estimated from a parametric bootstrap.}
    \label{app:h1-1-experiments:tab:bright-state-results-cont-3}
\end{table}

We ran a version of the bright-state depumping experiment introduced in Ref.~\cite{Gae21} in order to independently quantify the strength of MCM-induced crosstalk errors in the H1-1 system. The H1-1 system operates on qubits encoded in the atomic hyperfine states of $\textsuperscript{171}\mathrm{Yb}^{+}$, with the $\ket{1}$ state encoded in the ``bright'' state, named because an ion in the bright state will fluoresce photons when measured. A bright-state depumping experiment quantifies MCM-induced crosstalk errors by measuring the depumping rates (i.e., the rate of leakage from the $\ket{1}$ state to outside of the computational space) of unmeasured ions out of the bright state due to photon scattering off the measured ion. 

The bright-state depumping experiment introduced in Ref.~\cite{Gae21} was designed to measure crosstalk errors within each gate zone, so we ran a modified experiment to simultaneously quantify crosstalk errors in each gate and auxiliary zone. Our modified bright-state depumping experiment consisted of multiple rounds of the following: 
\begin{enumerate}
    \item H1-1's 20 ions were arranged as in Fig.~\ref{fig:qirb-combined}(g) and individually prepared in the bright state $\ket{1}$. 
    \item A sequence of multiple measurement pulses were applied to all of the even or odd numbered ions in the gate zones.
    \item Spin-flip transitions [$X(\pi)$ pulses] were applied to all of the unmeasured ions, ideally leaving each of the unmeasured ions in the dark state $\ket{0}$.
    \item The unmeasured ions were measured.
\end{enumerate}
Ideally, if no depumping occurred, we would measure each of the unmeasured ions in the dark state $\ket{0}$. However, if an unmeasured ion interacted with a scattered photon, then it would only be returned to the dark state $\ket{0}$ with probability approximately $1/3$. Under the assumption of unpolarized light, the decay of the bright-state population is 
\begin{equation}
    p_{\mathrm{depump}}(t) = \frac{2}{3}(1-e^{-3\gamma t}),
\end{equation}
where $\gamma$ is the scattering rate and $t$ is the duration of the first measurement pulse~\cite{Gae21}. We extracted $\gamma$ by fitting the decay in the bright-state population as a function of $t$. The depumping rate of the ions in the auxiliary zones is the average depumping rate from the alternating rounds of measuring even and odd-numbered ions, while the depumping rate of the gate zone ions is just the depumping rate obtained from measuring the other parity ions. For added interpretability, we converted the depumping rate of each ion into a process infidelity. Moreover, we constructed $1\sigma$ error bars from a parametric bootstrap. See Tables~\ref{app:h1-1-experiments:tab:bright-state-results}-\ref{app:h1-1-experiments:tab:bright-state-results-cont-3} for complete results. 

\section{\algiers\ calibration data}\label{app:ibmq-algiers-calibration-data}
\begin{table*}[h!]
\begin{tabular}{|l|l|l|l|l|l|l|l|l|}
\cline{1-9}
qubit & $T_1$ (us) & $T_2$ (us) & frequency (GHz)& anharmonicity  (GHz) & readout error & Pr(prep 1, measure 0)& Pr(prep 0, measure 1)& readout length  (ns) \\\cline{1-9}
\Q0 & 162.10 & 70.21 & 4.95 & -0.34 & 0.007 & 0.009 & 0.006 & 910.22 \\\cline{1-9}
\Q1 & 128.00 & 107.33 & 4.84 & -0.35 & 0.007 & 0.011 & 0.004 & 910.22 \\\cline{1-9}
\Q2 & 163.25 & 318.37 & 5.05 & -0.34 & 0.006 & 0.006 & 0.005 & 910.22 \\\cline{1-9}
\Q3 & 57.59 & 34.71 & 5.20 & -0.34 & 0.010 & 0.009 & 0.012 & 910.22 \\\cline{1-9}
\Q4 & 74.06 & 195.70 & 4.96 & -0.34 & 0.005 & 0.008 & 0.002 & 910.22 \\\cline{1-9}
\Q5 & 154.49 & 151.87 & 5.12 & -0.34 & 0.007 & 0.007 & 0.006 & 910.22 \\\cline{1-9}
\Q6 & 110.83 & 255.31 & 4.99 & -0.34 & 0.009 & 0.012 & 0.006 & 910.22 \\\cline{1-9}
\Q7 & 109.13 & 169.52 & 4.88 & -0.35 & 0.009 & 0.010 & 0.007 & 910.22 \\\cline{1-9}
\Q8 & 152.72 & 358.81 & 5.02 & -0.33 & 0.008 & 0.011 & 0.006 & 910.22 \\\cline{1-9}
\Q9 & 120.42 & 157.36 & 4.90 & -0.35 & 0.008 & 0.013 & 0.004 & 910.22 \\\cline{1-9}
\Q10 & 146.66 & 210.42 & 4.95 & -0.34 & 0.020 & 0.026 & 0.014 & 910.22 \\\cline{1-9}
\Q11 & 232.92 & 31.18 & 5.10 & -0.34 & 0.010 & 0.013 & 0.008 & 910.22 \\\cline{1-9}
\Q12 & 136.63 & 63.91 & 5.06 & -0.34 & 0.010 & 0.012 & 0.008 & 910.22 \\\cline{1-9}
\Q13 & 110.36 & 161.33 & 5.11 & -0.34 & 0.025 & 0.025 & 0.025 & 910.22 \\\cline{1-9}
\Q14 & 121.83 & 87.96 & 4.99 & -0.34 & 0.011 & 0.013 & 0.009 & 910.22 \\\cline{1-9}
\Q15 & 134.57 & 160.62 & 4.81 & -0.34 & 0.012 & 0.015 & 0.009 & 910.22 \\\cline{1-9}
\Q16 & 77.11 & 9.83 & 4.84 & -0.34 & 0.060 & 0.036 & 0.085 & 910.22 \\\cline{1-9}
\Q17 & 100.54 & 54.18 & 4.89 & -0.35 & 0.008 & 0.012 & 0.005 & 910.22 \\\cline{1-9}
\Q18 & 109.56 & 37.46 & 5.01 & -0.34 & 0.017 & 0.016 & 0.018 & 910.22 \\\cline{1-9}
\Q19 & 120.88 & 48.68 & 5.22 & -0.33 & 0.026 & 0.028 & 0.024 & 910.22 \\\cline{1-9}
\Q20 & 79.05 & 66.00 & 4.96 & -0.34 & 0.125 & 0.035 & 0.216 & 910.22 \\\cline{1-9}
\Q21 & 172.59 & 33.31 & 5.06 & -0.34 & 0.020 & 0.025 & 0.015 & 910.22 \\\cline{1-9}
\Q22 & 143.22 & 145.92 & 4.98 & -0.34 & 0.007 & 0.009 & 0.006 & 910.22 \\\cline{1-9}
\Q23 & 111.20 & 45.22 & 5.21 & -0.34 & 0.018 & 0.012 & 0.023 & 910.22 \\\cline{1-9}
\Q24 & 171.32 & 66.29 & 5.02 & -0.34 & 0.006 & 0.006 & 0.005 & 910.22 \\\cline{1-9}
\Q25 & 111.33 & 57.33 & 4.90 & -0.34 & 0.015 & 0.018 & 0.011 & 910.22 \\\cline{1-9}
\Q26 & 76.27 & 38.58 & 4.84 & -0.31 & 0.063 & 0.056 & 0.071 & 910.22 \\\cline{1-9}
\end{tabular}
\caption{\textbf{IBMQ Algiers calibration data.} Calibration data from \algiers\ from the time of our QIRB demonstrations. }
    \label{tab:ibmq_algiers_calibration}
\end{table*}
\end{document}